\documentclass[table,english,envcountsame,runningheads]{llncs}
\usepackage[utf8]{luainputenc}
\usepackage[T1]{fontenc}
\usepackage{xspace}
\usepackage{tcolorbox}
\usepackage{amsmath}
\usepackage{amssymb}
\usepackage{amsfonts}
\usepackage{marvosym}
\usepackage{bm}
\usepackage{hyperref}
\usepackage{paralist}
\usepackage[inline]{enumitem}
\usepackage[smaller]{acronym}
\usepackage{threeparttable}
\usepackage[noend]{algpseudocode}
\usepackage{cryptocode}
\algrenewcommand\algorithmicensure{\textbf{Parties:}} 
\algrenewcommand\algorithmicrequire{\textbf{Input:}}
\algrenewcommand\algorithmicindent{2em}%
\usepackage[plain]{algorithm}

\usepackage[framemethod=TikZ]{mdframed}

\usepackage{microtype}
\spnewtheorem*{lemma*}{Lemma}{\bfseries}{\itshape}
\spnewtheorem*{theorem*}{Theorem}{\bfseries}{\itshape}

\newcommand{\cC}{\mathcal{C}}

\newcommand{\cF}{\mathcal{F}}

\DeclareMathOperator{\algh}{\algo{H}}
\DeclareMathOperator{\algp}{\algo{P}}

\DeclareMathOperator{\binary}{bin}

\DeclareTextCommand{\textdollar}{T1}{\scalebox{.85}{\symbol{`\$}}}
\DeclareTextCommand{\mathdollar}{T1}{\scalebox{.60}{\symbol{`\$}}}
\newcommand{\AlgReturn}[1]{ \textbf{return} #1}
\newcommand{\algo}[1]{ \mathsf{#1}}
\newcommand{\event}{\mathcal{E}}
\DeclareMathOperator{\query}{\mathsf{query}}

\DeclareMathOperator{\WOR}{*}

\DeclareMathOperator{\PKSteg}{\mathsf{PKStS}}
\DeclareMathOperator{\Steg}{\mathsf{PKStS}}
\DeclareMathOperator{\PKES}{\mathsf{PKES}}
\DeclareMathOperator{\Enc}{\mathsf{Enc}}
\DeclareMathOperator{\Dec}{\mathsf{Dec}}

\DeclareMathOperator{\Ward}{\mathsf{W}}
\DeclareMathOperator{\Find}{\mathsf{Find}}
\DeclareMathOperator{\Guess}{\mathsf{Guess}}
\DeclareMathOperator{\Gen}{\mathsf{Gen}}
\DeclareMathOperator{\sscca}{ss-cca}

\DeclareMathOperator{\ml}{\mathsf{ml}}
\DeclareMathOperator{\setup}{\mathsf{Setup}}
\DeclareMathOperator{\cl}{\mathsf{cl}}
\DeclareMathOperator{\outl}{\mathsf{ol}}
\DeclareMathOperator{\fin}{\mathsf{in}}
\DeclareMathOperator{\fout}{\mathsf{out}}
\DeclareMathOperator{\dl}{\mathsf{dl}}
\DeclareMathOperator{\Eval}{\mathsf{Eval}}
\DeclareMathOperator{\Fi}{\mathsf{Fi}}
\DeclareMathOperator{\Att}{\mathsf{A}}
\DeclareMathOperator{\SSCCA-Dist}{\mathsf{SS-CCA-Dist}}

\newcommand{\pk}{\textit{pk}}
\newcommand{\sk}{\textit{sk}}
\newcommand{\combine}[2]{#1\!.\!#2}

\newcommand{\sgets}{\twoheadleftarrow}
\newcommand{\minent}{H_{\infty}}
\DeclareMathOperator{\Dist}{\mathsf{Dist}}

\DeclareMathOperator{\adv}{\mathbf{Adv}}
\DeclareMathOperator{\Adv}{\mathbf{Adv}}

\DeclareMathOperator{\cca}{cca}
\DeclareMathOperator{\ccad}{cca\text{\$}}
\DeclareMathOperator{\hash}{hash}

\DeclareMathOperator{\prp}{prp}

\DeclareMathOperator{\pkess}{pkes}
\DeclareMathOperator{\supp}{supp}
\DeclareMathOperator{\Ins}{InS}
\DeclareMathOperator{\Sgen}{Sgen}
\DeclareMathOperator{\Nq}{Nq}
\DeclareMathOperator{\Nsui}{Nsui}

\DeclareMathOperator{\unrel}{\mathbf{UnRel}}

\DeclareMathOperator{\rejsam}{\mathsf{rejsam}}
 \DeclareMathOperator{\hist}{\mathsf{hist}}
 \DeclareMathOperator{\lex}{lex}
\newcommand{\ie}{i.\,e.\xspace}

\newcommand{\eg}{e.\,g.\xspace}

\DeclareMathOperator{\negl}{\mathsf{negl}}
\DeclareMathOperator{\chan}{\mathcal{C}}
\DeclareMathOperator{\Perm}{\mathsf{Perms}}

\newacro{CRHF}{collision-resistant hash function}
        \newacro{PPTM}{polynomial probabilistic Turing machine}

        \newacro{PRP}{pseudorandom permutation}
                \newacro{PRF}{pseudorandom function}

        \newacro{PKES}{public key encryption scheme}

        \newacro{CCA}{chosen-ciphertext attack}
        \newacro{RCCA+}[RCCA\$]{replayable chosen-covertext\$
          attack}
        \newacro{CCA+}[CCA\$]{chosen-ciphertext\$ attack}

        \newacro{SS-CCA}{steganographic chosen-covertext attack}
        \newacro{SS-RCCA}{steganographic replayable chosen-covertext
          attack}
                \newacro{SES}{symmetric encryption scheme}

\newcommand{\myAlgorithm}[4]{
  \begin{algorithm}[H]
  \begin{tcolorbox}[title={#4: #1}, colframe=black!10, coltitle=black]
    \begin{algorithmic}[1]
\vspace*{-1mm}
\Require #2
      #3
    \end{algorithmic}
  \end{tcolorbox}
\end{algorithm}
\vspace{-.8cm}
}

\title{On the Gold Standard for Security of\\ Universal Steganography}

\author{Sebastian Berndt\inst{1} \and Maciej Li\'{s}kiewicz\inst{2}}

\institute{Department of Computer Science, Kiel University\\
  \email{seb@informatik.uni-kiel.de}\and
  Institute for Theoretical Computer Science, University of L\"{u}beck\\
 \email{liskiewi@tcs.uni-luebeck.de}}

 \authorrunning{S. Berndt and M. Li\'{s}kiewicz}

\begin{document}

\maketitle\thispagestyle{headings}

\begin{abstract}
  While symmetric-key steganography is quite well understood both in the
  information-theoretic and in the computational setting, many fundamental
  questions about its public-key counterpart resist persistent attempts to solve
  them. The computational model for public-key steganography was proposed by von
  Ahn and Hopper in EUROCRYPT~2004. At TCC~2005, Backes and Cachin gave the
  first universal public-key stegosystem~--~\ie one that works on all
  channels~--~achieving security against replayable chosen-covertext attacks
  (\acs{SS-RCCA}) and asked whether security against non-replayable
  chosen-covertext attacks (\acs{SS-CCA}) is achievable. Later, Hopper
  (ICALP~2005) provided such a stegosystem for every efficiently sampleable
  channel, but did not achieve universality. He posed the question whether
  universality \emph{and} \acs{SS-CCA}-security can be achieved simultaneously.
  No progress on this question has been achieved since more than a decade. In
  our work we solve Hopper's problem in a somehow complete manner: As our main
  positive result we design an \acs{SS-CCA}-secure stegosystem that works for
  \emph{every} memoryless channel. On the other hand, we prove that this result
  is the best possible in the context of universal steganography. We provide a
  family of \emph{$0$-memoryless} channels -- where the already sent documents
  have only marginal influence on the current distribution -- and prove that no
  \acs{SS-CCA}-secure steganography for this family exists in the standard
  non-look-ahead model.
\end{abstract}



\section{Introduction}
Steganography is the art of hiding the transmission of information
to achieve secret communication without revealing its presence. 
In the basic setting, the aim of the steganographic encoder 
(often called Alice or the stegoencoder) is to hide a secret 
message in a document and to send it to the stegodecoder (Bob)
via a public \emph{channel} which is completely monitored 
by an \emph{adversary} (Warden or steganalyst). The channel is
modeled as a probability distribution of legal documents, 
called \emph{covertexts}, and the adversary’s task is to 
distinguish those from altered ones, called \emph{stegotexts}.
Although strongly connected with cryptographic encryption, 
steganography is not encryption: While encryption only tries to hide 
the content of the transmitted message, steganography aims 
to hide both the message and the fact that a message was 
transmitted at all. 

As in the cryptographic setting, the security of the stegosystems
should only rely on the secrecy of the keys used by the system. 
\emph{Symmetric-key} steganography, which assumes that Alice 
and Bob share a secret-key, has been a subject of intensive study 
both in an information-theoretic \cite{Cachin2004,ryabko2011constructing,wang2008perfectly} 
and in a computational setting 
\cite{dedic2009upper,hopper2002provably,hopper2009provably,katzenbeisser2002defining,kiayias2014stego,liskiewicz2013grey}.
A drawback of such an approach is that 
the encoder and the decoder must have shared a key 
in a secure way. This may be unhandy, 
\eg if the encoder communicates with several 
parties. 

In order to avoid this problem in cryptography, Diffie and
Hellman provided the notion of a \emph{public-key scenario} in their groundbreaking work \cite{diffie1976publickey}.
This idea has proved to be
very useful and is currently used in nearly every cryptographic
application. Over  time, the notion of security against so-called
\emph{chosen ciphertext attacks} (\ac{CCA}-security) has established itself
as the ``gold standard'' for security in the public-key scenario
\cite{hofheinz2016standard,kiltz2010adaptive}. In
this setting, an attacker has also access to a decoding oracle that
decodes every ciphertext different from the challenge-text.  Dolev,
Dwork and Naor \cite{dolev2000cca} proved that the simplest assumption
for public-key cryptography -- the existence of trapdoor permutations -- is
sufficient to construct a \ac{CCA}-secure public key cryptosystem.

Somewhat in contrast  to the research in cryptographic encryption, only very 
little studies in steganography have been concerned so far within the
public-key setting. Von Ahn and Hopper \cite{vonAhn2003public,vonAhn2004public}
were the first to give a formal framework 
and to prove that secure public-key steganography exists.  
They formalized security against a \emph{passive} adversary 
in which Warden is allowed to provide challenge-hiddentexts to Alice
in hopes of distinguishing covertexts from stegotexts 
encoding the hiddentext of his choice. For a restricted model,
they also defined  security against an active adversary; It is 
assumed, however, that Bob must know the identity of Alice, 
which deviates from the common bare public-key scenario.

Importantly, the schemes provided in \cite{vonAhn2003public,vonAhn2004public} 
are \emph{universal} (called also \emph{black-box} in the literature).
This property guarantees that the systems are secure with 
respect not only to a concrete channel 
$\cC$ 
but to a broad range of channels. 
The importance of universality is based on the fact that
typically no good description of the distribution
of a channel is known.

In \cite{backes2005active}, Backes and Cachin provided a notion of
security for public-key stegano\-graphy with \emph{active} attacks,
called \emph{\acp{SS-CCA}}.  In this scenario the warden may provide a
challenge-hiddentext to Alice and enforce the stegoencoder to send
stegotexts encoding the hiddentext of his choice.  The warden may then
insert documents into the channel between Alice and Bob and observe
Bob's responses in hope of detecting the steganographic communication.
This is the steganographic equivalent of a chosen ciphertext attack
against encryption and it seems to be the most general type of security
for public-key steganography with active attacks similar to
\ac{CCA}-security in encryption.  Backes and Cachin also gave a
universal public-key stego\-system which, although not secure in the
general \ac{SS-CCA}-setting, satisfies a relaxed notion called
\emph{steganographic security against publicly-detectable replayable
  adaptive chosen-covertext attacks} (\ac{SS-RCCA}) inspired by the work
of Canetti et al.~\cite{canetti2003relaxing}. In this relaxed setting, 
the warden may still provide a hiddentext to Alice and is allowed to
insert documents into the channel between Alice and
Bob 
but with the restriction that the warden's document does not encode the
chosen hiddentext. Backes and Cachin left as an open problem if secure
public-key steganography exists at all in the \ac{SS-CCA}-framework.

This question 
was answered by Hopper \cite{hopper2005covertext} in the affirmative in case
Alice and Bob communicate via an efficiently sampleable channel $\chan$. He
proved (under the assumption of a \ac{CCA}-secure cryptosystem) that for every
such channel $\chan$ there is an \ac{SS-CCA}-secure stegosystem $\Steg_{\chan}$
on $\chan$. The system cleverly ``derandomizes''~sam\-pling documents by using
the sampling-algorithm of the channel and using a pseudo\-random generator to
deterministically embed the encrypted message. Hence, $\Steg_{\chan}$ is only secure
on the single channel $\chan$ and is thus not universal. Hopper
\cite{hopper2005covertext} posed as a challenging open problem to show the
(non)existence of a universal \ac{SS-CCA}-secure stegosystem.
Since more than a decade,
public key steganography has been used as a tool in different contexts
(\eg broadcast steganography \cite{fazio2014broadcast} and private
computation \cite{chandran2007covert,cho2016covert}), but 
this fundamental question remained open. 

We solve Hopper's problem in a complete manner by proving (under the
assumption of the existence of doubly-enhanced trapdoor permutations and
collision-resistant hash functions) the existence of an
\ac{SS-CCA}-secure public key stegosystem that works for every
\emph{memoryless} channel, \ie such that the documents are
independently distributed (for a formal definition see next section).
On the other hand, we also prove that the influence of the history --
the already sent documents -- dramatically limits the security of
stegosystems in the realistic non-look-ahead model: We show that no stegosystem
can be \ac{SS-CCA}-secure against all $0$-memoryless channels in the
non-look-ahead model. In these channels, the influence of the history is
minimal. We thereby
demonstrate a clear dichotomy result for universal public-key 
steganography: While memoryless
channels do exhibit an \ac{SS-CCA}-secure stegosystem, the introduction
of the history prevents this kind of security.

\subsubsection*{Our Contribution.}
As noted above, the stegosystem of Backes and Cachin has the drawback
that it achieves a weaker security than \ac{SS-CCA}-security while it
works on every channel \cite{backes2005active}. On the other hand, the
stegosystem of Hopper achieves \ac{SS-CCA}-security but is specialized
to a single channel~\cite{hopper2005covertext}.  We prove (under the
assumption of the existence of doubly-enhanced trapdoor permutations and
collision-resistant hash functions) that there is a stegosystem that is
\ac{SS-CCA}-secure on a large class of channels (namely the memoryless
ones). The main technical novelty is a method to generate covertexts for
the message $m$ such that finding a second sequence of covertexts that
encodes $m$ is hard. Hopper achieves this at the cost of the
universality of his system, while we still allow a very large class of
channels. We thereby answer the question of Hopper in the affirmative,
in case of memoryless channels.  Note that before this work, it was not
even known whether an \ac{SS-CCA}-secure stegosystem exists that works
for some class of channels (Hopper's system only works on a single
channel that is hard-wired into the system). Furthermore, we prove that
\ac{SS-CCA}-security for
memoryless channels is the best possible in a very natural model: If the
history influences the channel distribution in a minor way, \ie only by
its length, we prove that \ac{SS-CCA}-security is not achievable in the
standard non-look-ahead model of von Ahn and Hopper. In
Table~\ref{table:one}, we compare our results with previous works.

\renewcommand{\arraystretch}{1.5}
\begin{table}[ht] 
 \centering
  \small
  \begin{threeparttable}[b]
  \caption{Comparison of the public-key stegosystems}
  \begin{tabular}{lccc}
    Paper & Security & Channels& Applicability\\
    \hline    \hline
    von Ahn and Hopper \cite{vonAhn2003public} \hspace*{6mm} & passive & universal & possible\\
    \hline    Backes and Cachin \cite{backes2005active} &
                \ac{SS-RCCA} & universal& possible\\
    \hline    Hopper \cite{hopper2005covertext} & \ac{SS-CCA} & single constr. channel& possible\\
    \hline    This work (Theorem \ref{thm:sec}) & \ac{SS-CCA} & \quad all memoryless channels\quad& possible\\ 
    \hline    This work (Theorem \ref{thm:impossibility})& \ac{SS-CCA} &
                                                                         universal
                               & \quad impossible\tnote{*} \\ \hline
  \end{tabular}
           \begin{tablenotes}
         \item [*] In the non-look-ahead model against non-uniform wardens.
         \end{tablenotes}
\label{table:one}
\end{threeparttable}
\end{table} 

\subsubsection*{Related Results.}
Anderson and Petitcolas \cite{anderson1998limits} and Craver \cite{craver1998public}, 
have both, even before the publication of the work by von Ahn and Hopper \cite{vonAhn2003public,vonAhn2004public},
described ideas for public-key steganography, however,  with only 
heuristic arguments for security.
Van Le and Kurosawa \cite{vanLe2006public} showed that every
efficiently sampleable channel has an \ac{SS-CCA}-secure public-key
stegosystem. A description of the channel is built into the
stegosystem and it makes use of a pseudo-random generator $\algo{G}$ that 
encoder and decoder share. But the authors make a strong 
assumption concerning changes of internal states of $\algo{G}$ 
each time the embedding operation is performed, which does not fit into
the usual models of cryptography and steganography.
Lysyanskaya and Meyerovich \cite{lysyanskaya2006imperfect} investigated
the influence of the sampling oracle on the security of public key
stegosystems with passive attackers. They prove that the stegosystem of
von Ahn and Hopper \cite{vonAhn2004public} becomes insecure if the
approximation of the channel distribution by the sampling oracle 
deviates only slightly from the correct distribution. They also construct a
channel, where no incorrect
approximation of the channel yields a secure stegosystem. This
strengthens the need for universal stegosystems, as even tiny
approximation errors of the channel distribution may lead to huge
changes with regard to the security of the system.
Fazio, Nicolosi and Perera \cite{fazio2014broadcast} extended public-key
steganography to the multi-recipient setting, where a single sender communicates
with a dynamically set of receivers. Their system is designed such
that no outside party and no unauthorized user is able to detect the
presence of these broadcast communication.
Cho, Dachma-Soled and Jarecki \cite{cho2016covert} upgraded the covert
multi-party computation model of Chandran et
al.~\cite{chandran2007covert} to the concurrent case and gave 
protocols for several fundamental operations, \eg
string equality and set intersection. Their steganographic (or
\emph{covert}) protocols are based upon the decisional
Diffie-Hellman  problem.

\vspace*{2mm}\noindent
The paper is organized as follows.
Section~\ref{sec:model-steganography} contains the basic definitions
 and notations.
In Section
\ref{sec:detect}, we give an example attack on the stegosystem of Backes
and Cachin to highlight the differences  between
\ac{SS-RCCA}-security and \ac{SS-CCA}-security. The following
Section~\ref{sec:an-high-level} contains a high-level view of our
construction. 
Section~\ref{sec:bias} uses the results of~\cite{hopper2005covertext}
to prove that one can construct cryptosystems with
ciphertexts that are indistinguishable from a distribution on bitstrings
related to the hypergeometric distribution, which we will need later
on. The main core of our protocol is an algorithm to order the documents in an
undetectable way that still allows us to transfer information. This
ordering is described in Section~\ref{sec:ordering}. 
Our results concerning the existence of \ac{SS-CCA}-secure
steganography for every memoryless channel are then presented and proved
in Section~\ref{sec:stegosystem}. Finally,
Section~\ref{sec:impossibility} contains the impossibility result for
\ac{SS-CCA}-secure stegosystems in the non-look-ahead model on
$0$-memoryless channels.

In order to improve the presentation, 
we moved proofs of some technical statements 
to the appendix.

\section{Definitions and Notation}
\label{sec:model-steganography}
If $S$ is a finite set, we write
$x\sgets S$ to denote the \emph{random} assignment of a uniformly
chosen element of $S$ to $x$. If $A$ is a probability distribution or a randomized
algorithm, we write $x\gets A$ to denote the assignment of the
output of $A$, taken over the internal
coin-flips of~$A$. 

As our cryptographic and steganographic primitives will be parameterized
by the key length $\kappa$, we want that the ability of any polynomial
algorithm to attack this primitives is lower than the inverse of all
polynomials in $\kappa$. This is modeled by the definition of a
negligible function.  A function $\negl\colon \mathbb{N}\to [0,1]$ is
called \emph{negligible}, if for every polynomial $p$, there is an
$N_{0}\in \mathbb{N}$ such that $\negl(N) < p(N)^{-1}$ for every
$N\geq N_{0}$.
For a probability distribution $D$ on support $X$, the
\emph{min-entropy} $\minent(D)$ is defined as $\inf_{x\in X}\{-\log
D(x)\}$. 

We also need the notion of a \emph{strongly 2-universal hash function},
which is a set of functions $G$ mapping bitstrings of length $\ell$ to
bitstrings of length $\ell' < \ell$ such that for all
$x,x'\in \{0,1\}^{\ell}$ with $x\neq x'$ and all (not necessarily
different) $y,y'\in \{0,1\}^{\ell'}$, we have
$|\{f\in G\mid f(x)=y \land f(x')=y'\}|=\frac{|G|}{2^{2\ell'}}$. If
$\ell/\ell'\in \mathbb{N}$, a
typical example of such a family is the set of functions 
\begin{center}
  $\{x\mapsto
  \left(\sum_{i=1}^{\ell/\ell'}a_{i}x_{i}+b\right)\bmod 2^{\ell'}
  \mid a_{1},\ldots,a_{\ell/\ell'},b\in \{0,\ldots,2^{\ell'}-1\}\ \},$
\end{center}
where
$x_{i}$ denotes the $i$-th block of length $\ell'$ of $x$ and we
implicitly use the canonical bijection between $\{0,1\}^{n}$ and the finite field
$\{0,\ldots,2^{n}-1\}$. 
See \eg the
textbook of Mitzenmacher and Upfal
\cite{mitzenmacher2005probability} for more information on this.
For two polynomials $\ell$ and $\ell'$, a \emph{strongly 2-universal
  hash family} is a family $\mathcal{G}=\{G_{\kappa}\}_{\kappa\in \mathbb{N}}$ such
that every $G_{\kappa}$ is a strongly 2-universal hash function mapping
strings of length $\ell(\kappa)$ to strings of length $\ell'(\kappa)$. 

\vspace*{-1mm}
\subsubsection*{Channels and  Stegosystems.}
\label{sec:steg-steg}
In order to be able to embed messages into unsuspicious communication,
we first need to provide a definition for this. 
We model the
communication as an unidirectional transfer of \emph{documents} that we
will treat as strings of length~$n$ over a constant-size alphabet~$\Sigma$. 
The communication is defined via the concept of a \emph{channel} $\chan$ on
$\Sigma$: A function, that maps, for
every $n\in \mathbb{N}$,  a \emph{history} $\hist\in (\Sigma^{n})^{*}$ to a
probability distribution on $\Sigma^{n}$. We denote this probability
distribution by $\mathcal{C}_{\hist,n}$ and its \emph{min-entropy}
$\minent(\chan,n)$ as $\min_{\hist}\{\minent(\chan_{\hist,n})\}$. 
\begin{definition}
  We say that a channel $\cC$ is
  \emph{memoryless}, if 
  $\cC_{\hist,n}=\cC_{\hist',n}$ for all 
  $\hist,\hist'$, \ie if the history has no effect on the channel distribution.
\end{definition}
Note  the difference between memoryless and 
\emph{$0$-memoryless} channels of Lysyanskaya and Meyerovich
\cite{lysyanskaya2006imperfect}, where only the \emph{length} of the history
has an influence on the channel, since 
the channel distributions are
described by the use of  \emph{memoryless Markov chains}:
\begin{definition}[\cite{lysyanskaya2006imperfect}]
    A channel $\cC$ is \emph{$0$-memoryless}, if
    $\cC_{\hist,n}=\cC_{\hist',n}$ for all $\hist,\hist'$ such that
    $|\hist|=|\hist'|$.
  \end{definition}



 A stegosystem $\Steg$ tries to embed messages of length $\combine{\Steg}{\ml}$
 into $\combine{\Steg}{\outl}$  documents 
 of the channel $\chan$ that each have size $\combine{\Steg}{\dl}$, such that this
 sequence is indistinguishable from a
 sequence of typical documents.
  A \emph{public-key stegosystem}
$\Steg$ with
  message length $\combine{\Steg}{\ml}\colon \mathbb{N}\to \mathbb{N}$,
 document length
  $\combine{\Steg}{\dl}\colon \mathbb{N}\to \mathbb{N}$, and output length $\combine{\Steg}{\outl}\colon
  \mathbb{N}\to \mathbb{N}$ (all functions of the security parameter $\kappa$) is a triple of
  \acp{PPTM}  $[\combine{\Steg}{\Gen},\combine{\Steg}{\Enc},\combine{\Steg}{\Dec}]$\footnote{We
    will drop the prefix $\Steg$ if the context is clear.}  with the 
  functionalities:
  \begin{itemize}
  \item The \emph{key generation} 
    $\Gen$ on input $1^{\kappa}$ produces a pair
    $(\pk,\sk)$ consisting of a \emph{public
      key} $\pk$ and a  \emph{secret key} $\sk$ 
    (we assume
    that $\sk$ also fully contains $\pk$).
  \item The \emph{encoding} algorithm $\Enc$ takes as input the public key
    $\pk$, a message $m\in \{0,1\}^{\ml(\kappa)}$, a history
    $\hist\in (\Sigma^{\dl(\kappa)})^{*}$ and some state information $s\in
    \{0,1\}^{*}$ and produces a document $d\in \Sigma^{\dl(\kappa)}$ and state
    information $s'\in \{0,1\}^{*}$ by being able to sample from 
    $\mathcal{C}_{\hist,\dl(\kappa)}$. By $\Enc^{\cC}(\pk,m,\hist)$, we denote
    the complete output of $\outl(\kappa)$ documents one by one.
Note that generally, the encoder needs to decide upon document $d_{i}$
before it is able to get samples for the $(i+1)$-th document, as in the
secret-key model of Hopper et al.~\cite[Section 2, ``channel
access'']{hopper2009provably} and the public-key model of von Ahn and
Hopper~\cite[Section 3]{vonAhn2003public,vonAhn2004public}. This
captures the notion that an attacker should have as much information as
possible 
while the stegosystem is not
able to look-ahead into the future. 
To highlight this restriction, we call
this model the \emph{non-look-ahead model}. Note that this is no
restriction for memoryless channels. 
  \item The \emph{decoding} algorithm $\Dec$ takes as input the secret key
    $\sk$,  a sequence of documents $d_{1},\ldots,d_{\outl(\kappa)}$, history $\hist$ and
    outputs a message $m'$. 
  \end{itemize}

\vspace*{2mm}
\noindent 
The following properties are essential for stegosystems $\Steg$ with
output length $\ell=\combine{\Steg}{\outl}(\kappa)$. 
It is \emph{universal} (\emph{black box}), if it works on
every channel without prior knowledge of the probability distribution of
the channel. Clearly channels with too small min-entropy (such as deterministic
channels) are not suitable for steganographic purposes. We thus concentrate only
on channels with sufficiently large min-entropy.

The system is \emph{reliable} if the probability that the
  decoding fails is bounded by a negligible function. Formally, the
  \emph{unreliability} $\unrel_{\Steg,\chan}(\kappa)$ is defined as
  probability that the decoding fails, \ie
  \begin{align*}
  \max_{m,\hist}\{\Pr_{(\pk,\sk)\gets
\combine{\Steg}{\Gen}(1^{\kappa})}[\combine{\Steg}{\Dec}(\sk,\combine{\Steg}{\Enc}^{\chan}(\pk,m,\hist),\hist)\allowbreak\neq m]\}.    
  \end{align*}

The system $\Steg$ is \emph{secure}, if every polynomial
  attacker $\Ward$ (the \emph{warden}) has only negligible success
  probability. 
$\Ward$ works in two phases: In the first phase
  (called $\combine{\Ward}{\Find}$), the warden
has access to the channel $\chan$ and to a \emph{decoding oracle} $\Dec_{\sk}(\cdot)$, that returns
upon input $d_{1},\ldots,d_{\ell}$ and $\hist$ the same result as
$\combine{\Steg}{\Dec}(\sk,(d_{1},\ldots,d_{\ell}),\hist)$. At the end of the first phase, the
warden chooses a message $m^{*}$ and a history $\hist^{*}$. 

At the beginning
of the second phase (called $\combine{\Ward}{\Guess}$), the warden gets a sequence of documents
$\vec{d}^{*}=d_{1}^{*},\ldots,d_{\ell}^{*}$, which is with probability of $50\%$ the
result of $\combine{\Steg}{\Enc}^{\cC}(\pk,m^{*},\hist^{*})$ and with probability of $50\%$ just
the result of sampling $\ell$ random documents from $\cC_{\hist^{*},n}$,
which we denote as $\chan_{\hist^{*},n}^{\ell}$. Speaking more precisely, 
this sampling is done as follows: 
$d^{*}_i\gets \chan_{\hist^{*}\mid\mid d^{*}_1\mid\mid \ldots \mid\mid d^{*}_{i-1},n}$
for $i=1,\ldots,\ell$.
Next, the warden should distinguish between these two cases by
having access to another \emph{decoding oracle}.
 Depending on the oracle type, 
two definitions 
for security were proposed  by Backes and Cachin in~\cite{backes2005active}.

In the \ac{SS-CCA}-model (chosen covertext attack), 
the decoding oracle, denoted as $\Dec_{\sk,\vec{d}^{*}}(\cdot)$,
works like $\Dec_{\sk}(\cdot)$ on every input different from
$\vec{d}^{*}$. If $\Dec_{\sk,\vec{d}^{*}}(\cdot)$ is called with input $\vec{d}^{*}$, it
  simply returns $\bot$. 
In the weaker \ac{SS-RCCA}-model (restricted chosen ciphertext attack), 
the decoding oracle, denoted as 
$\Dec_{\sk,\vec{d}^{*},m^{*}}(\cdot)$, works like $\Dec_{\sk}(\cdot)$ on
most inputs: If the input $d$ equals $\vec{d}^{*}$ or is a
valid encoding of $m^{*}$ (a \emph{replay} of $\vec{d}^{*}$), the oracle simply returns~$\bot$. 

Formally, \ac{SS-CCA}-security is defined by the \ac{SS-CCA}-security game 
given below and  the advantage of
$\Ward=[\combine{\Ward}{\Find},\combine{\Ward}{\Guess}]$ is defined as
\begin{align*}
\adv^{\sscca}_{\Ward,\Steg,\chan}(\kappa) = \bigl|
\Pr[\SSCCA-Dist(\Ward,\Steg,\chan,\kappa)=1]-\frac{1}{2}\bigl|.
\end{align*} 

\vspace*{-6mm}

\myAlgorithm{$\SSCCA-Dist(\Ward,\Steg,\chan,\kappa)$}{
  warden $\Ward$, stegosystem $\Steg$, channel $\chan$, security
  parameter $\kappa$}{
\State $(\pk,\sk)\leftarrow \combine{\Steg}{\Gen}(1^{\kappa})$; $(m^{*},\hist^{*},s)\gets \combine{\Ward}{\Find}^{\Dec_{\sk},\chan}(\pk)$
\State  $b\gets \{0,1\}$
\State $\textbf{if }b=0\textbf{ then }\vec{d}^{*}\leftarrow
             \combine{\Steg}{\Enc}^{\cC}(\pk,m^{*},\hist^{*})\textbf{ else }\vec{d}^{*}\leftarrow \chan_{\hist^{*},n}^{\ell}$
     \State $b' \leftarrow \combine{\Ward}{\Guess}^{\Dec_{\sk,\vec{d}^{*}},\chan}(\pk,m^{*},\hist^{*},s,\vec{d}^{*})$
     \State $\textbf{if } b'=b \
              \textbf{then return} \ 1
             \textbf{ else return} \ 0$
           }
           {\ac{SS-CCA}-security game}
\vspace*{-2mm}\noindent
A stegosystem 
$\Steg$ is called \ac{SS-CCA}-secure against channel $\chan$ if for some 
{negligible} function $\negl$ and 
all wardens $\Ward$,  we have 
$\adv^{\sscca}_{\Ward,\Steg,\chan}(\kappa)\le \negl(\kappa)$. We define 
\ac{SS-RCCA}-security analogously, where the $\Guess$ phase uses $\Dec_{\sk,\vec{d}^{*},m^{*}}$ 
as decoding oracle. Formally, a stegosystem is \emph{universally \ac{SS-CCA}-secure} (or
just universal), if it is \ac{SS-CCA}-secure against all channels of
sufficiently large (\ie super-logarithmic in $\kappa$) min-entropy. 

\vspace*{-2mm}
\subsubsection*{Cryptographic Primitives.} 
\label{sec:crypt-prim}
Due to space constraints, we only give informal definitions of the used 
cryptographic primitives and refer the reader to the textbook
of Katz and Lindell \cite{katz2007introduction} for complete definitions.

We will make use of different cryptographic primitives, namely hash
functions, pseudorandom permutations and \ac{CCA}-secure cryptosystems.
 A \emph{\ac{CRHF}}
  $\algh=(\combine{\algh}{\Gen},\combine{\algh}{\Eval})$ is a pair of
  \acp{PPTM} such that $\combine{\algh}{\Gen}$ upon input
 $1^{\kappa}$
  produces a key $k\in \{0,1\}^{\kappa}$. The keyed function
  $\combine{\algh}{\Eval}$ takes the key
  $k\gets \combine{\algh}{\Gen}(1^{\kappa})$ and a string
  $x\in \{0,1\}^{\combine{\algh}{\fin}(\kappa)}$ and produces a string
  $\combine{\algh}{\Eval}_{k}(x)$ of length
  $\combine{\algh}{\fout}(\kappa) < \combine{\algh}{\fin}(\kappa)$.
  The probability of every \acs{PPTM} $\Fi$ to
   find a collision -- two strings $x\neq x'$ such that
  $\combine{\algh}{\Eval}_{k}(x)=\combine{\algh}{\Eval}_{k}(x')$ -- upon
  random choice of $k$ is negligible. For a set $X$, denote by $\Perm(X)$ the set of all
  permutations on $X$.  A \emph{\ac{PRP}}
  $\algp=(\combine{\algp}{\Gen},\combine{\algp}{\Eval})$ is a pair of
  \acp{PPTM} such that $\combine{\algp}{\Gen}$ upon input $1^{\kappa}$
  produces a key $k\in \{0,1\}^{\kappa}$. The keyed function
  $\combine{\algp}{\Eval}$ takes the key
  $k\gets \combine{\algp}{\Gen}(1^{\kappa})$ and is a permutation on the
  set $\{0,1\}^{\combine{\algp}{\fin}(\kappa)}$. An attacker $\Dist$
  (the \emph{distinguisher}) is given black-box access to $P\sgets
  \Perm(\{0,1\}^{\combine{\algp}{\fin(\kappa)}})$ or to $\combine{\algp}{\Eval}_{k}$ for
  a randomly chosen $k$ and should distinguish between those
  scenarios. The success probability of every $\Dist$ is negligible.
A \emph{\ac{PKES}}
$\PKES=(\combine{\PKES}{\Gen},\combine{\PKES}{\Enc},\allowbreak\combine{\PKES}{\Dec})$
is a triple of \acp{PPTM}
such that $\combine{\PKES}{\Gen}(1^{\kappa})$ produces a pair of keys $(\pk,\sk)$ with
$|\pk| = \kappa$ and $|\sk| =  \kappa$. The key $\pk$ is called the
\emph{public key} and the key $\sk$ is called the \emph{secret key}
(or \emph{private key}). 
The
\emph{encryption algorithm} $\combine{\PKES}{\Enc}$ takes as input 
 $\pk$ and a plaintext $m\in
\{0,1\}^{\combine{\PKES}{\ml}(\kappa)}$ of length
$\combine{\PKES}{\ml}(\kappa)$ and
outputs a ciphertext $c\in \{0,1\}^{\combine{\PKES}{\cl}(\kappa)}$ of
length $\combine{\PKES}{\cl}(\kappa)$. The \emph{decryption
  algorithm} $\combine{\PKES}{\Dec}$ takes as input  $\sk$ and the
ciphertext $c$ and produces a plaintext $m\in
\{0,1\}^{\combine{\PKES}{\ml}(\kappa)}$. Informally, we will allow an
attacker $\Att$ to
first choose a message $m^{*}$ that should be encrypted and denote this
by $\combine{\Att}{\Find}$. In the next step ($\combine{\Att}{\Guess}$), the
attacker gets $c^{*}$, which is either $\Enc(\pk,m^{*})$ or a random
bitstring. He is allowed to decrypt ciphertexts different from $c^{*}$
and his task is to distinguish between these two cases.
This security notion is known as security
against \acp{CCA+}. For an attacker $\Att$ on cryptographic
primitive $\Pi\in \{\hash,\prp,\pkess\}$ with implementation $X$, we
write $\Adv^{\Pi}_{\Att,X,\chan}(\kappa)$ for the success probability of
$\Att$ against $X$ relative to channel $\chan$, \ie the attacker $\Att$
also has access to a sampling oracle of $\chan$. In case of encryption schemes,
the superscript $\ccad$ is used instead of $\pkess$.

Due to the
works~\cite{dolev2000cca,goldreich2013enhancements,luby1985permutations,naor1989hash}
we know that \ac{CCA+}-secure cryptosystems and \acp{PRP} can be
constructed from doubly-enhanced trapdoor permutations resp.~one-way
functions, while \acp{CRHF} can not be constructed from them in a
\emph{black-box} way, as Simon showed an
oracle-separation in~\cite{simon1998collision}.

\section{Detecting the Scheme of Backes and Cachin}
\label{sec:detect}
In order to understand the difference between \ac{SS-CCA}-security and the
closely related, but weaker, \ac{SS-RCCA}-security, we give a short
presentation of the universal \ac{SS-RCCA}-stegosystem of Backes and
Cachin \cite{backes2005active}. We also show
that their system is not \ac{SS-CCA}-secure, which was already noted by Hopper
in \cite{hopper2005covertext}. 
The proof of insecurity 
nicely illustrates the difference between the
security models. It 
also highlights the main difficulty of
\ac{SS-CCA}-security: One needs to prevent so called \emph{replay
  attacks}, where the warden constructs upon 
stegotext $c$ another stegotext $c'$~--~the \emph{replay} of $c$~--~that
embeds the same message as $c$. 

Backes and Cachin \cite{backes2005active} showed that there is a
universal \ac{SS-RCCA}-secure stegosystem under the assumption that a
\ac{RCCA+}-secure cryptosystem exists.\footnote{The definition of a \ac{RCCA+}-secure 
cryptosystem is analogous to \ac{SS-RCCA}-security 
given in Section \ref{sec:steg-steg}.} They make 
use of 
a technique 
called \emph{rejection} \emph{sampling}. Let $\{G_{\kappa}\}_{\kappa\in \mathbb{N}}$ be a strongly 2-universal
hash function family, $f\in G_{\kappa}$ a function, $\cC$ be a
channel, $\hist$ be a history and $b\in \{0,1\}$ be a bit. The algorithm 
$\rejsam(f,\cC,b,\hist)$ samples documents
$d\leftarrow \cC_{\hist,\dl(\kappa)}$ until it finds a document $d^{*}$ such
that $f(d^{*})=b$ or until it has sampled $\kappa$ documents.
 If
$\PKES$ is an \ac{RCCA+}-secure 
cryptosystem, 
they define a
stegosystem that computes $(b_{1},\ldots,b_{\ell})\leftarrow \combine{\PKES}{\Enc}(\pk,m)$
and then sends $d_{1},d_{2},\ldots,d_{\ell}$, where
$d_{i}\leftarrow \rejsam(f,\cC,b_{i},
\hist || d_{1} || \ldots || d_{i-1})$.
The function $f\in G_{\kappa}$ is also part of the public key. The system
is universal as it does not assume any knowledge on~$\cC$.

They then prove that this stegosystem is \ac{SS-RCCA}-secure. And indeed,
one can show that their stegosystem is not \ac{SS-CCA}-secure by constructing
a generic warden 
$\Ward$  that
works as follows: The first phase $\combine{\Ward}{\Find}$ chooses as message
$m^{*}=00\cdots0$ and as 
$\hist^{*}$
the empty history $\varnothing$. The second phase
$\combine{\Ward}{\Guess}$ gets 
$\vec{d}^{*}=d^{*}_{1},\ldots,d^{*}_{\ell}$ which is either a sequence of
 random documents or the output of the stegosystem on $\pk$, $m^{*}$, and
$\hist^{*}$.  The warden $\Ward$ now computes another document
$d'$ via rejection sampling that embedds $f(d^{*}_{\ell})$ 
(the \emph{replay} of $\vec{d}^{*}$) and decodes
$d^{*}_{1},\ldots,d^{*}_{\ell-1},d'$ via the decoder of the rejection
sampling stegosystem. It then
returns $0$ if the returned message $m'$ consists only of zeroes. If $\vec{d}^{*}$ was a sequence
of  random documents, it is highly unlikely that $\vec{d}^{*}$ decodes
to a message that only consists of zeroes. If $\vec{d}^{*}$ was produced by
the stegosystem, the decoder only returns
something different from the all-zero-message if $d'=d^{*}_{\ell}$ which
is highly unlikely. The warden $\Ward$ has advantage of
$1-\negl(\kappa)$ and the stegosystem is thus not
\ac{SS-CCA}-secure. Backes and Cachin posed the question
whether a universal \ac{SS-CCA}-secure stegosystem exists.

\section{An High-Level View of our Stegosystem}
\label{sec:an-high-level}
The stegosystem of Backes and Cachin only achieves \ac{SS-RCCA}-security
as a single ciphertext has many different possible encodings in terms of
the documents used.  Hopper achieves \ac{SS-CCA}-security by limiting
those encodings: Due to the sampleability of the channel, each
ciphertext has exactly one deterministic encoding in terms of the
documents. While Hopper achieves \ac{SS-CCA}-security, he needs to give
up the universality of the stegosystem, as a description of the channel is
hard-wired into the stegosystem. In order to handle as many
channels as possible, we will allow many different encodings of the same
ciphertext, but make it hard to find them for anyone but the
stegoencoder. To simplify the presentation, we focus on the case
of embedding a single bit per document. Straightforward modifications
allow embedding of $\log(\kappa)$ bits.

Our stegosystem, named  $\PKSteg^{*}$ will use the following
approach to encode a message $m$: It first samples, for sufficiently large $N$,  
a set $D$ of $N$ documents  from the channel
$\chan$ and uses a strongly $2$-universal hash function
$f\in G_{\kappa}$ to split these documents into documents $D_0$ that encode
bit $0$ (\ie $D_{0}=\{d\in D\mid f(d)=0\}$) and $D_1$ that
encode bit~$1$ (\ie $D_{1}=\{d\in D\mid f(d)=1\}$). 
Now we encrypt the message $m$ via a certain public-key
encryption system, named $\PKES^{\WOR}$ (described in the next section), and obtain a ciphertext
$\vec{b}=b_{1},\ldots,b_{L}$ of length $L=\lfloor N/8 \rfloor$. 
Next our goal is to order the
documents in $D$ into a sequence $\vec{d}=d_{1},\ldots,d_{N}$ such that
the first $L$ documents $d_{1},\ldots,d_{L}$ encode $\vec{b}$ (\ie
$f(d)_{i}=b_{i}$). This ordering is performed by the algorithm
$\algo{generate}$. 
However, the attacker
still has several possibilities for a replay attack on this scheme, for example:  
\begin{itemize}
\item He could exchange some document $d_{i}$ by another document $d'_{i}$ with
  $f(d_{i})=f(d'_{i})$ (as $f$ is publicly known) and the sequence
  $d_{1},\ldots,d_{i-1},d'_{i},d_{i+1},\ldots,d_{N}$ would be a replay of
  $\vec{d}$. Such attacks will be called \emph{sampling attacks}. To prevent the
  attacker from exchanging a sampled document by a
  non-sampled one, we also encode a hash-value of all sampled documents $D$ and
  transmit this hash value to Bob. 
\item The attacker can exchange documents $d_{i}$ and $d_{j}$, with $i < j$ and
  $f(d_{i})=f(d_{j})$, and the resulting sequence
  $d_{1},\ldots,d_{i-1},d_{j},d_{i+1},\ldots,d_{j-1},d_{i},d_{j+1},\ldots,d_{N}
  $ would be a replay of $\vec{d}$. Such attacks will be called \emph{ordering
    attacks}. We thus need to prevent the attacker from
  exchanging the positions of sampled documents. We achieve this by
  making sure that the ordering of the documents generated by $\algo{generate}$
  is deterministic, \ie for each set of documents $D$ and each ciphertext
  $\vec{b}$, the ordering $\vec{d}$ generated by $\algo{generate}$ is
  deterministic. This property is achieved by using \acp{PRP} to sort the
  sampled documents $D$. The corresponding keys of the \acp{PRP} are also
  transmitted to Bob and the stegodecoder can thus also compute this
  deterministic ordering. 
\end{itemize}
In total, our stegoencoder $\combine{\PKSteg^{*}}{\Enc}$ works on a secret message $m$ and 
on a publicly known hash-function $f$ as follows:
\begin{enumerate}
\item Sample $N$ documents $D$ from the channel;
\item Get a hash-key $k_{\algh}$ and compute a hash-value
  $h=\combine{\algh}{\Eval}_{k_{\algh}}(\lex(D))$ of the sampled documents,
  where $\lex(D)$ denotes the sequence of elements of $D$ in lexicographic
  order. This prevents sampling attacks, where a sampled document is replaced by a
  non-sampled one;
\item Get two\footnote{We believe that one permutation suffices. But in order
    to improve the readability of the proof for security, we use two
    permutations in our stegosystem.} \ac{PRP}-keys $k_{\algp}$ and $k'_{\algp}$
  that will be used to determine the unique ordering of the documents in $D$ via
  $\algo{generate}$. This prevents ordering attacks, where the order of the
  sampled documents is switched;
\item Encrypt the concatenation of $m, k_{\algh}, k_{\algp}, k'_{\algp}, h$ via
  a certain \acl{PKES} $\PKES^{\WOR}$ and obtain the ciphertext $\vec{b}$ of
  length $L=\lfloor N/8 \rfloor$. As long as $\PKES^{\WOR}$ is secure, the
  stegodecoder is thus able to verify whether all sampled documents were sent
  and can also verify the ordering of the documents.
\item Compute the ordering $\vec{d}$ of the documents $D$ via $\algo{generate}$
  that uses the \ac{PRP} keys $k_{\algp}$ and $k'_{\algp}$ to determine the
  ordering of the documents. It also uses the ciphertext $\vec{b}$ to guarantee
  that the first $L$ send documents encode the ciphertext $\vec{b}$, \ie 
        $b_1\ldots b_L=f(d_1)\ldots f(d_L)$;
\item Send the ordering of the documents $\vec{d}$. 
\end{enumerate}

To decode a sequence  of documents $\vec{d} = d_1,\ldots,d_{N}$, the stegodecoder of
$\PKSteg^{*}$ computes the ciphertext $b_1=f(d_1),\ldots,b_{L}=f(d_{L})$ encoded
in the first $L$ documents of $\vec{d}$. It then decodes this 
ciphertext $b_1\ldots b_L$ via $\PKES^{\WOR}$ to obtain the message $m$, the
\ac{PRP} keys $k_{\algp}$ and $k'_{\algp}$, the hash-key $k_{\algh}$ and the
hash-value $h$. First it verifies the hash-value  by
checking whether $\combine{\algh}{\Eval}_{k_{\algh}}(\lex(\{d_1,\ldots,d_N\}))$
equals the hash-value $h$ to prevent sampling attacks. It then uses the \ac{PRP}
keys $k_{\algp}$ and $k'_{\algp'}$ to compute an ordering of the received
documents via $\algo{generate}$ to verify that no ordering attack was used. 
If these validations are successful, the decoder $\combine{\PKSteg^{*}}{\Dec}$
returns $m$;  Otherwise, it concludes that  $\vec{d}$ is not
a valid stegotext and returns~$\bot$.

Intuitively, it is clear that a successful sampling attack on this scheme would
break the \acl{CRHF} $\algh$, as it needs to create a collision of
$\lex(D)$ in order to pass the first verification step. Furthermore, a
successful ordering attack would need manipulate the ciphertext $\vec{b}$ and
thus break the security of the \acl{PKES} $\PKES^{\WOR}$, as
the \ac{PRP} keys $k_{\algp}$ and $k'_{\algp}$ guarantee a deterministic
ordering of the documents. 


As explained above, our stegoencoder computes the ordering $\vec{d}=
d_1,\ldots,d_{N}$ of the documents $D=\{d_1,\ldots,d_{N}\}$ via the
deterministic algorithm $\algo{generate}$, that is given the following
parameters: the set of documents $D$, the hash-function $f$ and the ciphertext
$\vec{b}$ to ensure that the first documents of the ordering encode $\vec{b}$.
It has furthermore access to the \ac{PRP} keys $k_{\algp}$ and $k'_{\algp}$ that
guarantee a deterministic ordering of the documents in $D$ and thus prevents
ordering attacks. As the ordering $\vec{d}$ produced by $\algo{generate}$ is
sent by the stegoencoder, this ordering must be indistinguishable from a random
permutation on $D$ (which equals the channel distribution) in order to be
undetectable. As $f(d_1)=b_1,\ldots, f(d_L)=b_L$, not every distribution upon
the ciphertext $\vec{b}$ can be used to guarantee that $\vec{d}$ is
indistinguishable from a uniformly random permutation. 
This indistinguishability is guaranteed by requiring that the
ciphertext $\vec{b}$ is distributed according to a certain distribution
corresponding to a random process modeled by drawing black and white balls from
an urn without replacement. In our setting, the documents in $D$ will play the
role of the balls and the coloring is given by the function $f$.

Section~\ref{sec:bias} describes this random process in
detail and proves that we can indeed construct a public-key encryption
system that produces ciphertexts that are indistinguishable from this 
process. Section~\ref{sec:ordering} contains a formal description of
$\algo{generate}$, proves that no attacker can produce a replay of its
output and shows that the generated permutation is indeed
indistinguishable from a random permutation. Finally,
Section~\ref{sec:stegosystem} contains the complete description of the
stegosystem.

\section{Obtaining Biased Ciphertexts}
\label{sec:bias}

We will now describe
a probability distribution and show how one can
derive a symmetric encryption scheme with ciphertexts that are
indistinguishable from this distribution. In order to do this, we first define a
channel that represents the required probability distribution together with appropriate
parameters, use Theorem~\ref{thm:hopper} to derive a stegosystem for this
channel, and finally derive a cryptosystem from this stegosystem.

Based upon a \ac{CCA+}-secure public-key cryptosystem $\PKES$,
Hopper~\cite{hopper2005covertext} constructs for every efficiently sampleable channel $\chan$ an
\ac{SS-CCA}-secure stegosystem $\PKSteg_{\chan}$ by ``derandomizing'' the
rejection sampling algorithm. The only requirement
upon the channel $\chan$ is the existence of the efficient sampling
algorithm and that the stegoencoder and the stegodecoder use the
same sampling algorithm. Importantly, due to the efficient sampleability of $\chan$,
the encoder of $\PKSteg_{\chan}$ does not need an access to the sample oracle.
Thus, we get the following result.

\begin{theorem}[Theorem~2 in \cite{hopper2005covertext}]
\label{thm:hopper}
If $\chan$ is an efficiently sampleable channel
and $\PKES$ is a \ac{CCA+}-secure
public-key cryptosystem 
(which can be constructed from doubly enhanced
trapdoor permutations\footnote{See \eg the work \cite{goldreich2013enhancements} of
  Goldreich and Rothblum.}) then there is a
stegosystem $\PKSteg_{\chan}$ (without an access to the sample oracle)
such that for all wardens $\Ward$ there is a negligible function $\negl$
such that
  \begin{align*}
    \Adv^{\sscca}_{\Ward,\PKSteg_{\chan},\chan}(\kappa)\leq
    \negl(\kappa)+ 2 ^{-\minent(\chan,\kappa)/2}.
  \end{align*}
\end{theorem}

Note that the system $\PKSteg_{\chan}$ is guaranteed to be secure (under the
assumption that \ac{CCA+}-secure public-key cryptosystems exist), if the
channel $\chan$ is efficiently sampleable and has min-entropy
$\omega(\log \kappa)$. We call such a channel \emph{suitable}.


The probability distribution for the ciphertexts we are interested in 
is the distribution for the bitstrings $\vec{b}$ we announced in the the previous section.
As we will see later, the required probability can be described equivalently as follows:
\begin{itemize}
\item We are given $N$ elements: $N_{0}$
         of them are labeled with $0$ and the remaining $N-N_{0}$
         elements are labeled with $1$. 
\item We draw randomly a sequence of $K$ elements from the set 
         (drawing without replacements) and 
         look at the generated bitstring $\vec{b}=b_{1}\ldots b_{K}$ of length $K$ 
         determined by the labels of the elements. 
\end{itemize}

%
%

We will assume that there are
enough elements of both types, \ie that $N_{0}\geq K$ and $N-N_{0}\geq K$. 
The resulting probability distribution, denoted as $D^{\WOR}_{(N,N_{0},K)}$, upon
bitstrings of length $K$ is then given as
\begin{align}\label{eq:hypergeometric}
\begin{split}
  &\Pr[D^{\WOR}_{(N,N_{0},K)}=b_{1}\ldots
    b_{K}]\ = \
  \frac{1}{\binom{K}{|\vec{b}|_0}}\cdot \frac{\binom{N_0}{|\vec{b}|_0}\cdot \binom{N-N_0}{K-|\vec{b}|_0}}{\binom{N}{K}}\ =\\
  &\bigl(\prod_{j=0}^{K-1}\frac{1}{N-j}\bigr) \cdot
  \bigl(\prod_{j=0}^{|\vec{b}|_{0}-1}[N_{0}-j] \bigr)
  \cdot \bigl(\prod_{j=0}^{|\vec{b}|_{1}-1}[N-N_{0}-j] \bigr), 
  \end{split}
\end{align}
where $|\vec{b}|_{0}$ denotes the number of zero bits in $\vec{b}=b_{1},\ldots,b_{K}$ and $|\vec{b}|_{1}$ the
number of one bits in $\vec{b}$. Note that the distribution on the number of
zeroes within such bitstrings is a hypergeometric distribution with
parameters $N$, $N_{0}$, and $K$. 

Now we will construct a channel $\chan^{\WOR}$ upon key parameter
$\kappa$ with document length $n=\dl(\kappa)=\kappa$. In the definition below, $\binary(x)_{y}$ denotes the
binary representation  of length exactly $y$ for the integer $x$.
\begin{itemize}
\item For the empty history $\varnothing$, let
  $\chan^{\WOR}_{\varnothing,\kappa}$ be the uniform distribution on all strings 
  $\binary(N)_{\lceil \kappa/2 \rceil}\binary(N_{0})_{\lfloor \kappa/2 \rfloor}$
  that range over all positive integers 
  $N,N_0\leq 2^{ \lfloor \kappa/2 \rfloor}$ such that $N\ge 8\kappa$ 
  and $1/3 \leq N_{0}/N \leq 2/3$ 
  (in our construction we need initially a stronger condition than just $N_{0}\geq \kappa$ and $N-N_{0}\geq \kappa$).
\item If the history is of the form
   $\hist'=\binary(N)_{\lceil \kappa/2 \rceil}\binary(N_{0})_{\lfloor
   \kappa/2 \rfloor}\hist$ for some $\hist\in \{0,1\}^{*}$
   then we consider two cases: if $|\hist|\le\frac{1}{8}N$ then the distribution
   $\chan^{\WOR}_{\hist',\kappa}$ equals 
   $D^{\WOR}_{(N-|\hist|,N_{0}-|\hist|_{0},\kappa)}$;
   Otherwise, \ie if $|\hist|> \frac{1}{8}N$ then $\chan^{\WOR}_{\hist',\kappa}$ equals the 
   uniform distribution over $\{0,1\}^{\kappa}$.
\end{itemize}

It is easy to see that the min-entropy
$\minent(\chan^{\WOR},n)=\min_{\hist'}\{\minent(\chan^{\WOR}_{\hist',n})\}$
of the channel $\chan^{\WOR}$ is obtained for the history 
$\hist'=\binary(N)_{\lceil \kappa/2 \rceil}\binary(N_{0})_{\lfloor \kappa/2 \rfloor} \hist$,
with $8\kappa\le N \le 2^{ \lfloor \kappa/2 \rfloor}$ and such that 
$(i)$ $N_0=\frac{1}{3}N$  and $\hist=00\ldots 0$ of length $\frac{1}{8}N-\kappa$ or 
$(ii)$ $N_0=\frac{2}{3}N$ and $\hist=11\ldots 1$ of length $\frac{1}{8}N-\kappa$. 
In the first case we get that the min-entropy of the distribution 
$\chan^{\WOR}_{\hist',n}$ is achieved on the bitstring $11\ldots 1$ of length $\kappa$
and in the second case on $00\ldots 0$ of length $\kappa$.
By Eq. \eqref{eq:hypergeometric} the probabilities to get
such strings are equal to each other and, since $\kappa\le N/8$, they  can be estimated as follows:
\[
  \prod_{j=0}^{\kappa-1}\frac{2N/3-j}{7N/8-\kappa - j}
  \ \le \ \left( \frac{2N/3}{7N/8-\kappa} \right)^{\kappa}\ \le\ \left( \frac{2N/3}{6N/8} \right)^{\kappa}
  \ = \ (8/9)^{\kappa}.
\]
Thus, we get that $\minent(\chan^{\WOR},n)\ge \kappa \log(9/8) $.

Moreover one can efficiently simulate the choice of $N, N_{0}$, the
sampling process of $D^{\WOR}_{(N,N_{0},\kappa)}$ and the uniform sampling
in $\{0,1\}^{\kappa}$. Therefore we can conclude
\begin{lemma}
  \label{lem:about:perm:chan:2}
  The channel $\chan^{\WOR}$ is suitable, \ie
  it is efficiently sampleable and has min-entropy $\omega(\log \kappa)$. 
  Furthermore, for history 
  $\hist=\binary(N)_{\lceil \kappa/2 \rceil}\binary(N_{0})_{\lfloor\kappa/2 \rfloor}$, with 
  $8\kappa\le N \le 2^{\lceil \kappa/2 \rceil}$ and 
  $1/3 \leq N_{0}/N \leq 2/3$, and for any integer $\ell \le \frac{N}{8 \kappa}$, 
  the bitstrings $\vec{b}=b_1\ldots b_K$ of length $K= \kappa \cdot \ell \le N/8$  
  obtained by the concatenation of $\ell$ consecutive  documents sampled from the channel
  with history $\hist$, \ie
  $b_i \gets \chan^{\WOR}_{\hist b_1 \ldots b_{i-1},n=\kappa}$, 
  have distribution  $D^{\WOR}_{(N,N_{0},K)}$.
\end{lemma}

A proof for the second statement of the lemma follows directly from the construction of 
the channel. 
Now, combining the first claim of the lemma  with
Theorem~\ref{thm:hopper} we get 
the following corollary.

\begin{corollary}
If doubly enhanced trapdoor permutations exists, 
there is a stegosystem $\PKSteg_{\chan^{\WOR}}$ (without an access to the sample oracle)
such that for all wardens $\Ward$ there is a negligible function $\negl$
such that
 $   \Adv^{\sscca}_{\Ward,\PKSteg_{\chan^{\WOR}},\chan^{\WOR}}(\kappa)\leq \negl(\kappa).$
\end{corollary}

Based upon this stegosystem $\PKSteg=\PKSteg_{\chan^{\WOR}}$, we construct a
public-key cryptosystem $\PKES^{\WOR}$,
with ciphertexts of length $\combine{\PKES^{\WOR}\!\!}{\cl}(\kappa) = \kappa\cdot \combine{\PKSteg}{\cl}(\kappa)$
 such
 that $\PKES^{\WOR}$ also has another algorithm, called
 $\combine{\PKES^{\WOR}\!\!}{\setup}$ that takes 
 parameters: two integers  $N$ and $N_{0}$ which satisfy
  $8\cdot\combine{\PKES^{\WOR}\!\!}{\cl}(\kappa)\le N \le 2^{ \lfloor \kappa/2 \rfloor}$
  and $N_{0}/N \in [1/3,2/3]$. Calling $\combine{\PKES^{\WOR}\!\!}{\setup}(N,N_0)$
  stores the values $N,N_0$ such that $\combine{\PKES^{\WOR}\!\!}{\Enc}$ and
  $\combine{\PKES^{\WOR}\!\!}{\Dec}$ can use them.
\begin{itemize}
\item The key generation $\combine{\PKES^{\WOR}\!\!}{\Gen}$ simply equals the
  key generation algorithm $\combine{\PKSteg}{\Gen}$.
\item The encoding algorithm $\combine{\PKES^{\WOR}\!\!}{\Enc}$ takes as
  parameters the public key $\pk$ and a message $m$. 
  It
  then simulates the encoder  $\combine{\PKSteg}{\Enc}$ on key $\pk$,
  message $m$ and history $\hist=\binary(N)_{\lceil \kappa/2 \rceil}\binary(N_{0})_{\lfloor \kappa/2
    \rfloor}$ and produces a bitstring
  of length $\combine{\PKES^{\WOR}\!\!}{\cl}(\kappa)=\combine{\PKSteg}{\outl}(\kappa)\cdot \kappa$. 
  \item The decoder $\combine{\PKES^{\WOR}\!\!}{\Dec}$ simply inverts this 
    process by simulating the stegodecoder $\combine{\PKSteg}{\Dec}$ on
    key $\sk$ and history $\hist=\binary(N)_{\lceil \kappa/2 \rceil}\binary(N_{0})_{\lfloor \kappa/2
    \rfloor}$. 
\end{itemize}

Clearly, the ciphertexts of $\combine{\PKES^{\WOR}\!\!}{\Enc}(\pk,m)$ are
indistinguishable from the distribution 
$D^{\WOR}_{(N,N_{0},\combine{\PKES^{\WOR}\!}{\cl}(\kappa))}$
by the second statement  of Lemma~\ref{lem:about:perm:chan:2}.  This
generalization of Theorem~\ref{thm:hopper} yields the
following corollary:
\begin{corollary}
\label{cor:indist:cca:perm}
  If doubly-enhanced trapdoor permutations exist, there is a secure public-key cryptosystem
  $\PKES^{\WOR}$,  equipped with the algorithm
  $\combine{\PKES^{\WOR}\!\!}{\setup}$ 
  that takes two parameters
  $N$ and $N_{0}$, such that its ciphertexts are
  indistinguishable from the probability distribution $D^{\WOR}_{(N,N_{0},\combine{\PKES^{\WOR}\!}{\cl(\kappa)}))}$
  whenever $N$ and $N_0$
  satisfy that  
  $8\cdot\combine{\PKES^{\WOR}\!\!}{\cl}(\kappa)\le N \le 2^{ \lfloor \kappa/2 \rfloor}$
  and $N_{0}/N \in [1/3,2/3]$.
\end{corollary}

\section{Ordering the Documents}
\label{sec:ordering}

As described before, to prevent replay attacks, we need to order the sampled
documents. This is done via the algorithm $\algo{generate}$ described in this
section. To improve the readability, we will abbreviate some terms and define
$L=\combine{\PKES^{\WOR}\!}{\cl}(\kappa)$ and
$n=\combine{\Steg^{*}\!}{\dl}(\kappa)$, where $\PKES^{\WOR}$ is the public-key
encryption scheme from the last section and $\Steg^{*}$ is our target
stegosystem that we will provide later on. We also define $N=8L$.

To order the set of documents $D\subseteq \Sigma^n$, we use the algorithm
$\algo{generate}$, presented below. It takes the set of documents
$D$ with $|D|=N$, a hash function $f\colon \Sigma^n \to \{0,1\}$ from
$G_{\kappa}$, a bitstring $b_1,\ldots,b_{L}$, and two
keys $k_{\algp}, k'_{\algp}$ for \ac{PRP}s.
 It then uses the \ac{PRP}s to find the right order of the documents.  

\myAlgorithm{$\algo{generate}(D,f,b_{1},\ldots,b_{L},k_{\algp},k'_{\algp})$}{
set $D$ with $|D|=N$, hash function $f$, bits $b_{1},\ldots,b_{L}$, \ac{PRP}-keys
$k_{\algp},k'_{\algp}$}{
    \State let $D_{0}=\{d\in D\mid f(d)=0\}$ and $D_{1}=\{d\in D\mid f(d)=1\}$
    \Comment{We assert that  $|D|=N$, and furthermore $|D_0|\in[N/3,2N/3]$}
  \For{$i=1$ to $L$}
   \State $d_{i} := \arg\min_{d\in
        D_{b_{i}}}\{\combine{\algp}{\Eval}_{k_{\algp}}(d)\}$;  $D_{b_{i}}$ := $D_{b_{i}}\setminus \{d_{i}\}$
      \EndFor
   \State let $D' = D_{0}\cup D_{1}$ \Comment{collect remaining documents}
   \For{$i=L+1,\ldots,N$}
   \State  $d_{i} := \arg\min_{d\in
     D'}\{\combine{\algp}{\Eval}_{k'_{\algp}}(d)\}$;  $D'$ :=
   $D'\setminus \{d_{i}\}$
    \EndFor
    \State \AlgReturn{$d_{1},d_{2},\ldots,d_{N}$}
    }{
    Algorithm}

    Note that the permutation $\combine{\algp}{\Eval}_{k_{\algp}}$ is a
    permutation upon the set $\{0,1\}^{n}$ (\ie on the documents themselves) and
    the canonical ordering of $\{0,1\}^{n}$ thus implicitly gives us an ordering
    of the documents. 

    We note the following important property of $\algo{generate}$ that
    shows where the urn model of the previous section comes into
    play. For  uniform random permutations $P$ and $P'$, we denote
    by $\algo{generate}(\cdots,P,P')$ the run of $\algo{generate}$,
    where the use of $\combine{\algp}{\Eval}_{k_{\algp}}$ is replaced by
    $P$ and the use of $\combine{\algp}{\Eval}_{k'_{\algp}}$ is replaced
    by $P'$. If the bits $\vec{b}=b_{1},\ldots,b_{L}$ are distributed
    according to $D^{\WOR}_{(N,|D_{0}|,L)}$, the resulting distribution
    on the documents then equals the channel distribution.
    \begin{lemma}
      \label{lem:infotheoretic}
      Let $\chan$ be any memoryless channel, $f$ be some hash function
      and $D$ be a set of $N=8L$ documents of $\chan$ such that
      $N/3\le |D_{0}| \le 2N/3$, where $D_{0}=\{d\in D\mid f(d)=0\}$. If
      the permutations $P,P'$ are uniformly random and the bitstring
      $\vec{b}=b_{1},\ldots,b_{L}$ is distributed according to
      $D^{\WOR}_{(N,|D_{0}|,L)}$, the
      output of $\algo{generate}(D,f,\vec{b},P,P')$ is a uniformly random
      permutation of $D$.
    \end{lemma}
    \begin{proof}
      Fix any document set $D$ of size $N=8L$ and a function $f$ that splits
      $D$ into $D_{0}\dot{\cup}D_{1}$, with $|D_{0}| \geq N/3$ and
      $|D_{1}|\geq N/3$. Let $\vec{\hat{d}}=\hat{d}_{1},\ldots,\hat{d}_{N}$ be any permutation on
      $D$. We will prove that the probability (upon bits $\vec{b}$ and
      permutations $P$, $P'$) that $\vec{\hat{d}}$ is produced, is $1/N!$ and thus
      establish the result. Let
      $\vec{d}=d_{1},\ldots,d_{N}$ be the random
      variables that denote the outcome of
      $\algo{generate}(D,f,b_{1},\ldots,b_{L},P,P')$. 

      Note that if $\vec{d}[i]$ (resp.~$\vec{\hat{d}}[i]$) denotes the
      prefix of length $i$ of $\vec{d}$ (resp.~$\vec{\hat{d}}$), 
      then using the chain rule formula we get
    \begin{align*}
      \Pr_{\vec{b},P,P'}[d_{1}d_{2}\ldots
        d_{N}=\hat{d}_{1}\hat{d}_{2}\ldots \hat{d}_{N}]\ = \
      \prod_{i=1}^{N}\Pr_{\vec{b},P,P'}[d_{i}=\hat{d}_{i} \mid \vec{d}[i-1]=\vec{\hat{d}}[i-1]].
    \end{align*}
    To estimate each of the factors of the product, we consider two cases:
    \begin{itemize}
    \item Case $i\leq L$: Let $\vec{\hat{b}}=\hat{b}_1,\ldots,\hat{b}_L$ be the
      bitstring such that 
  $\hat{b}_i=f(\hat{d}_i)$ and let $\vec{\hat{b}}[i]$ be the prefix
  $\hat{b}_1,\ldots,\hat{b}_i$ of $\vec{\hat{b}}$ of length $i$.
    Clearly, for $i\leq L$ it holds that the event $d_i=\hat{d}_{i}$ 
    under the condition $\vec{d}[i-1]=\vec{\hat{d}}[i-1]$ occurs iff (A) $d_{i}\in D_{\hat{b}_{i}}$ and
    (B) $d_{i}$ is put on position $|\vec{\hat{b}}[i]|_{\hat{b}_{i}}$ by the
    permutation $P$ with respect to $D_{\hat{b}_{i}}$. Due to the distribution
    of bit $b_{i}$ in the random bits $\vec{b}$, the event $d_{i}\in D_{\hat{b}_{i}}$ occurs with
    probability $(|D_{\hat{b}_{i}}|-|\vec{\hat{b}}[i-1]|_{\hat{b}_{i}})/(N-i+1)$  (under
    the above condition). As $\vec{d}[i-1]=\vec{\hat{d}}[i-1]$ holds,
    exactly $|\vec{\hat{b}}[i-1]|_{\hat{b}_{i}}$ documents from $D_{\hat{b}_{i}}$ are
    already used in the output. As $P$ is a uniform random permutation,
    the probability that $d_{i}$ is put on position $|\vec{\hat{b}}[i]|_{\hat{b}_{i}}$ by the
    permutation $P$ (with respect to $D_{\hat{b}_{i}}$) is thus
    $1/(|D_{\hat{b}_{i}}|-|\vec{\hat{b}}[i-1]|_{\hat{b}_{i}})$. Since (A) and (B) are independent,
    we conclude for $i\leq L$ that the probability 
    $\Pr_{\vec{b},P,P'}[d_i=\hat{d}_i\mid \vec{d}[i-1] = \vec{\hat{d}}[i-1]]$
    is equal to
    \begin{align*}
      &
     { \textstyle    \Pr_{\vec{b}}[d_{i}\in D_{\hat{b}_{i}}\mid \vec{d}[i-1]=\vec{\hat{d}}[i-1]] \ \times} \\
      &\quad \quad
      { \textstyle  \Pr_{P}[\text{$P$ puts $d_{i}$ on position $|\vec{\hat{b}}[i]|_{\hat{b}_{i}}$} \mid \vec{d}[i-1]=\vec{\hat{d}}[i-1]] \ = }\\
        &\frac{|D_{\hat{b}_{i}}| - |\vec{\hat{b}}[i-1]|_{\hat{b}_{i}}}{N-i+1} \cdot 
       \frac{1}{|D_{\hat{b}_{i}}| - |\vec{\hat{b}}[i-1]|_{\hat{b}_{i}}} \ = \ \frac{1}{N-i+1}.
    \end{align*}
\item Case $i> L$: 
  As the choice
  of $P'$ is independent from the choice
  of $P$, the
    remaining $2L$ items are ordered completely random. Hence,  for $i > L$ we also
    have 
    \begin{align*}
      \Pr_{\vec{b},P,P'}[d_{i}=\hat{d}_{i} \mid
      \vec{d}[i-1]=\vec{\hat{d}}[i-1]] \ = \ \frac{1}{N-i+1}.
    \end{align*}
    \end{itemize}
Putting it together, we get
    \begin{align*}
      \Pr_{\vec{b},P,P'}[d_{1}d_{2}\ldots
        d_{N}=\hat{d}_{1}\hat{d}_{2}\ldots \hat{d}_{N}] \ = \ 
      \prod_{i=1}^{N}\frac{1}{N-i+1}=\frac{1}{N!}.\tag*{\qed}
    \end{align*}
  \end{proof}

    As explained above, a second property that we need is that no attacker
    should be able to produce a
    ``replay'' of the output of $\algo{generate}$.
    Below, we formalize this notion and analyze
    the security of the algorithm. An attacker $\algo{A}$ on
    $\algo{generate}$ is a \ac{PPTM}, that receives nearly the same
    input as $\algo{generate}$: a set $D$ of $N$ documents, a hash
    function $f\colon \Sigma^n \to \{0,1\}$ from the family
    $G_{\kappa}$, a sequence $b_{1},\ldots,b_{L}$ of $L$ bits, and a key
    $k_{\algh}$ for the \ac{CRHF} $\algh$. Then $\algo{A}$ outputs a
    sequence $d'_{1},\ldots,d'_{N}$ of documents. We say that the
    algorithm $\algo{A}$ is \emph{successful} if
    \begin{enumerate}
    \item     $f(d_{i})=f(d'_{i})$ for all $i=1,\ldots,N$,
    \item
      $d'_{1},\ldots,d'_{N}=\algo{generate}(D',f,b_{1},\ldots,b_{L},k_{\algp},k'_{\algp})$,  and   
    \item $\combine{\algh}{\Eval}_{k_{\algh}}(\lex(D'))=\combine{\algh}{\Eval}_{k_{\algh}}(\lex(D))$,
    \end{enumerate}
    where $D'$ denotes the set $\{d'_{1},\ldots,d'_{N}\}$ and, recall,  $\lex(X)$
    denotes the sequence of elements of set $X$ in lexicographic
    order. We can then conclude the following lemma.

\begin{lemma}[Informal]
  \label{lem:produce_negl}
  Let $D\subseteq \Sigma^n$ be a set of documents with $|D|=N$, let
  $b_{1},\ldots,b_{L}$ be a bitstring, and $f\in G_{\kappa}$.  For every
  attacker $\algo{A}$ on $\algo{generate}$, there is a collision finder
  $\Fi$ for the \ac{CRHF} $\algh$ such that the probability that
  $\algo{A}$ is successful on $D,f,b_{1},\ldots,b_{L},k_{\algh}$ is
  bounded by $ \Adv_{\Fi,\algh,\chan}^{\hash}(\kappa)$.
\end{lemma}
The formal definition of ``$\algo{A}$ is successful'' as well as a
formal statement of the lemma can be found in the Appendix, Section
\ref{app:proofs:att:generate}.

\section{The Steganographic Protocol $\Steg^{*}$}
\label{sec:stegosystem}
We now have all of the ingredients of our stegosystem, namely the
\acs{CCA}-secure cryptosystem $\PKES^{\WOR}$ from Section~\ref{sec:bias} and the
ordering algorithm $\algo{generate}$ from Section~\ref{sec:ordering}. To improve
the readability, we will abbreviate some terms and define
$n=\combine{\Steg^{*}\!}{\dl}(\kappa)$,
$\ell=\combine{\Steg^{*}\!}{\outl}(\kappa)$, and
$L=\combine{\PKES^{\WOR}\!}{\cl}(\kappa)$, where $\PKES^{\WOR}$ is the
public-key encryption scheme from Section~\ref{sec:bias} and $\Steg^{*}$ is the
stegosystem that we will define in this section. We also let $N=8L$.

In the following, let $\chan$ be a memoryless channel, $\algp$ be a \ac{PRP}
relative to $\chan$, $\algh$ be a \ac{CRHF} relative to $\chan$ and
$\mathcal{G}=\{G_{\kappa}\}_{\kappa\in \mathbb{N}}$ be a strongly
$2$-universal hash family. 
Remember, that $\PKES^{\WOR}$ has the algorithm $\combine{\PKES^{\WOR}\!\!}{\setup}$
that takes the additional parameters
$N,N_{0}\leq 2^{\lceil  \kappa/2 \rceil}$, such that if 
$N \ge 8\cdot\combine{\PKES^{\WOR}\!\!}{\cl}(\kappa)$ and $N_{0}/N \in [1/3,2/3]$
then  the output of
$\combine{\PKES^{\WOR}\!\!}{\Enc}(\pk,m)$ is indistinguishable from
$D^{\WOR}_{(N,N_{0},\combine{\PKES^{\WOR}\!}{\cl}(\kappa))}$ (see Section
\ref{sec:bias} for a discussion). Furthermore, we assume that
$\PKES^{\WOR}$ has very sparse support, \ie the ratio of valid ciphertexts
compared to $\{0,1\}^{\combine{\PKES^{\WOR}\!}{\cl}(\kappa)}$ is
negligible: If $\combine{\PKES^{\WOR}\!}{\Enc}(\pk,m)$ is called, we
first use some public key encryption scheme $\PKES$ with very sparse
support to compute $c\gets \combine{\PKES}{\Enc}(\pk,m)$ and then
encrypt $c$ via $\PKES^{\WOR}$. This construction is due to
Lindell~\cite{lindell2003cca} and also maintains the
indistinguishability of the output of $\combine{\PKES^{\WOR}\!}{\Enc}$ and
the distribution $D^{\WOR}$, as this properties hold for all fixed
messages $m$.
Now we are ready to provide our stegosystem named $\PKSteg^{*}$. 
Its main core is the ordering algorithm $\algo{generate}$.
\begin{itemize}
\item 
The key generating $\combine{\PKSteg^{*}}{\Gen}$ queries
$\combine{\PKES^{\WOR}\!}{\Gen}$ for a key-pair $(\pk,\sk)$ and chooses a
hash-function $f\sgets G_{\kappa}$. The public key of the stegosystem will
be $\pk^{*}=(\pk,f)$ and the secret key will be $\sk^{*}=(\sk,f)$.

\item The encoding
algorithm $\combine{\PKSteg^{*}}{\Enc}$ presented below (as $\chan_n$ is
memoryless we skip $\hist$ in the description) works as described in
Section~\ref{sec:an-high-level}: 
It chooses appropriate keys, samples documents $D$, computes  a hash value of $D$,
generates bitstring $\vec{b}$ via $\PKES^{\WOR}$,
and finally orders
the documents via $\algo{generate}$. \footnote{That the number of
  produced documents is always divisible by $8$ does not hurt the
  security: The warden always gets the same number of documents, whether
  steganography is used or not.}

\item To decode a sequence of documents $d_1,\ldots,d_{N}$, the
  stegodecoder 
  $\combine{\PKSteg^{*}}{\Dec}$ first computes the bit string
  $b_1=f(d_1),\ldots,b_{N}=f(d_{N})$ and computes the number
   $N_{0}=|\{d_{i}\colon f(d_{i})=0\}|$.
In case $|\{ d_1,\ldots,d_{N}\}|< N$ or $N_{0}/N\not\in [1/3,2/3]$,
the decoder $\combine{\PKSteg^{*}}{\Dec}$ returns
$\bot$ and halts.
Otherwise, using $\combine{\PKES^{\WOR}\!}{\Dec}$ 
with $\sk$ and parameters $N, N_{0}$, 
it decrypts from the ciphertext $b_1,b_2, \ldots, b_{L}$ the message $m$,  the keys
$k_{\algh}, k_{\algp}, k'_{\algp}$ and the hash-value $h$. 
It then checks whether 
the hash-value $h$ is correct and
whether  
$d_1,\ldots,d_{N}=\algo{generate}(\{d_{1},\ldots,d_{N}\},f,b_{1},\ldots,b_{L},k_{\algp},k'_{\algp})$.
Only if this is
the case, the message $m$ is returned. Otherwise, 
$\combine{\PKSteg^{*}}{\Dec}$ decides that it can not
decode the documents and returns $\bot$.

\end{itemize}

\myAlgorithm{$\combine{\PKSteg^{*}}{\Enc}(\pk^{*},m)$
}{public key $\pk^{*}=(\pk,f)$, message $m$; access to channel $\chan_{n}$}{
\State let $L=\combine{\PKES^{\WOR}\!}{\cl}(\kappa)$ and $N=8L$;
 let $D_{0}:=\emptyset$ and $D_{1}:=\emptyset$
 \For {$j=1$ to $N$} 
   \State sample $d_{j}$ from $\chan_{n}$;
     let $D_{f(d_{j})}$ := $D_{f(d_{j})}\cup \{d_{j}\}$
\EndFor
    \State $N_{0}=|D_{0}|$
  \State \algorithmicif\ $|D_{0}\cup D_{1}|< N$ or $N_0/N\not\in [1/3,2/3]$\ \algorithmicthen\
  \AlgReturn{$d_{1},\ldots,d_{N}$} and \textbf{halt} 
  \State choose hash key $k_{\algh}\gets \combine{\algh}{\Gen}(1^{\kappa})$
  \State choose \ac{PRP} keys $k_{\algp},k'_{\algp}\gets \combine{\algp}{\Gen}(1^{\kappa})$
  \State let $h$ := $\combine{\algh}{\Eval}_{k_{\algh}}(\lex(D_{0}\cup
  D_{1}))$  \Comment{compute hash}
  \State call $\combine{\PKES^{\WOR}\!\!}{\setup}(N,N_0)$ \Comment{setup $N, N_0$}
  \State let $b_{1},b_{2},\ldots, b_{L} \gets \combine{\PKES^{\WOR}\!}{\Enc}(\pk,m\mid\mid k_{\algh}\mid\mid
  k_{\algp} \mid\mid k'_{\algp}\mid\mid h)$
\State let $\vec{d} := \algo{generate}(D_{0}\cup
D_{1},f,b_{1},\ldots,b_{L},k_{\algp},k'_{\algp})$ 
\State \AlgReturn{$\vec{d}$}
}{The steganographic encoder}
\subsubsection*{Proofs of Reliabiliy and Security.}
We will first concentrate on the reliability of the system $\PKSteg^{*}$ and prove
that its unreliability is negligible. This is due to the fact, that the
decoding always works and the encoding can only fail if a document was
drawn more than once or if the sampled documents are very imbalanced
with regard to $f$. 

\begin{theorem}
  \label{thm:rel}
  The probability that a message is not correctly embedded by the
  encoder $\combine{\PKSteg^{*}}{\Enc}$ 
  is at most
  $3N^{2}\cdot 2^{-\minent(\chan,\kappa)}+2\exp(-N/54)$.
\end{theorem}

If $1<\lambda \leq \log(\kappa)$ bits per document are embedded, this
probability is bounded by $2^{2\lambda}\cdot 3N^{2}\cdot
2^{-\minent(\chan,\kappa)}+2^{\lambda+1}\exp(-N/54)$, which is
negligible in $\kappa$ if $\minent(\chan,\kappa)$ sufficiently large. 
Now,
it only remains to prove that our construction is secure. The proof
proceeds similar to the security proof of Hopper
\cite{hopper2005covertext}. But instead of showing that no other
encoding of a message exists, we prove that finding any other encoding
of the message is infeasible via Lemma~\ref{lem:produce_negl}.

    \begin{figure}
      \begin{minipage}{.48\linewidth}
    \begin{gameproof}[name=$H$,arg={=\chan^{N}_n}]
      \gameprocedure[linenumbering]{
        \pk^{*}=(\pk,f) \gets \combine{\PKSteg^{*}}{\Gen}(1^{\kappa})\\
      \pcfor j := 1,2,\ldots,N:\\
      \t d_{j} \gets \chan_{\dl(\kappa)}\\
      \pcreturn ((d_{1},\ldots,d_{N}),\pk^{*})
    }
\end{gameproof}
      \end{minipage}
      \begin{minipage}{.48\linewidth}
\begin{gameproof}[name=$H$, arg={}, nr=1]
    \gameprocedure{
      \pk^{*}=(\pk,f) \gets \combine{\PKSteg^{*}}{\Gen}(1^{\kappa})\\
      \text{Lines 1 to 4 in
        $\combine{\PKSteg^{*}}{\Enc}$}\\
      \setcounter{pclinenumber}{4}
      \pcln \gamechange{$P\sgets \Perm$}\\
      \pcln  \pcreturn ((\gamechange{$d_{P(1)},\ldots,d_{P(N)}$}),\pk^{*})
    }
  \end{gameproof}
      \end{minipage}\\[.3cm]

  \begin{gameproof}[name=$H$, arg={}, nr=2]
        \gameprocedure{
      \pk^{*}=(\pk,f) \gets \combine{\PKSteg^{*}}{\Gen}(1^{\kappa})\\
      \text{Lines 1 to 4 in
        $\combine{\PKSteg^{*}}{\Enc}$}\\
      \setcounter{pclinenumber}{4}
      \pcln P\sgets \Perm; \gamechange{$P'\sgets \Perm$}; \gamechange{$k_{\algh}\gets \combine{\algh}{\Gen}(1^{\kappa})$}\\
      \pcln \gamechange{$b_{1},b_{2},\ldots, b_{L} \gets
        D^{\WOR}_{(N,N_0,L)}$}\\
      \pcln \pcreturn (\gamechange{$
      \algo{generate}(D_{0}\cup D_{1},f,b_{1},\ldots,b_{L},P,P')$},\pk^{*})\\
      \pccomment{$\algo{generate}(\ldots,P,P')$ uses the permutations $P,P'$}
    }
  \end{gameproof}\\[.3cm]

  \begin{gameproof}[name=$H$, arg={},nr=3]
    \gameprocedure{
      \pk^{*}=(\pk,f) \gets \combine{\PKSteg^{*}}{\Gen}(1^{\kappa})\\
      \text{Lines 1 to 4 in
        $\combine{\PKSteg^{*}}{\Enc}$}\\
      \setcounter{pclinenumber}{4}
      \pcln \gamechange{$k_{\algp} \gets
        \combine{\algp}{\Gen}(1^{\kappa})$}; P'\sgets \Perm; k_{\algh}\gets \combine{\algh}{\Gen}(1^{\kappa})\\
      \pcln b_{1},b_{2},\ldots, b_{L} \sgets D^{\WOR}_{(N,N_0,L)}\\
      \pcln \pcreturn
      (\algo{generate}(D_{0}\cup
      D_{1},f,b_{1},\ldots,b_{L},\gamechange{$k_{\algp}$},P'),\pk^{*})\\
      \pccomment{$\algo{generate}(\ldots,P')$ uses the permutation $P'$}
      }
    \end{gameproof}\\[.3cm]

      \begin{gameproof}[name=$H$, arg={},nr=4]
    \gameprocedure{
      \pk^{*}=(\pk,f) \gets \combine{\PKSteg^{*}}{\Gen}(1^{\kappa})\\
      \text{Lines 1 to 4 in
        $\combine{\PKSteg^{*}}{\Enc}$}\\
      \setcounter{pclinenumber}{4}
      \pcln k_{\algp} \gets
        \combine{\algp}{\Gen}(1^{\kappa}); \gamechange{$k'_{\algp} \gets
        \combine{\algp}{\Gen}(1^{\kappa})$}; k_{\algh}\gets \combine{\algh}{\Gen}(1^{\kappa})\\
      \pcln b_{1},b_{2},\ldots, b_{L} \sgets D^{\WOR}_{(N,N_0,L)}\\
      \pcln \pcreturn (\algo{generate}(D_{0}\cup
      D_{1},f,b_{1},\ldots,b_{L},k_{\algp},\gamechange{$k'_{\algp}$}),\pk^{*})
      }
    \end{gameproof}\\[.3cm]

      \begin{gameproof}[name=$H$, arg={=\combine{\PKSteg^{*}}{\Enc}},nr=5]
        \gameprocedure{
          \pk^{*}=(\pk,f) \gets \combine{\PKSteg^{*}}{\Gen}(1^{\kappa})\\
      \text{Lines 1 to 4 in
        $\combine{\PKSteg^{*}}{\Enc}$}\\
      \setcounter{pclinenumber}{4}
      \pcln k_{\algp} \gets \combine{\algp}{\Gen}(1^{\kappa}); k'_{\algp} \gets
        \combine{\algp}{\Gen}(1^{\kappa}); k_{\algh}\gets
        \combine{\algh}{\Gen}(1^{\kappa})\\
        \pcln \gamechange{$h := \combine{\algh}{\Eval}_{k_{\algh}}(\lex(D_0 \cup
          D_1))$}\\
        \pcln \gamechange{$\combine{\PKES^{\WOR}\!\!}{\setup}(N,N_0)$}\\
      \pcln b_{1},b_{2},\ldots, b_{L} \gets \gamechange{$\combine{\PKES^{\WOR}}{\Enc}(\pk,m\mid\mid k_{\algh}\mid\mid
  k_{\algp}\mid\mid k'_{\algp} \mid\mid h)$}\\
      \pcln \pcreturn (\algo{generate}(D_{0}\cup
      D_{1},f,b_{1},\ldots,b_{L},k_{\algp},k'_{\algp}),\pk^{*})
      }
  \end{gameproof}
      \setlength\overfullrule{8mm}
    \caption[An overview of hybrids $H_{1}$ and $H_{6}$ used in the proof of 
    Theorem~\ref{thm:sec}]{An overview of hybrids $H_{1}$
      and $H_{6}$  used in the proof of Theorem~\ref{thm:sec}.\\
      Changes between the  hybrids are marked as shadowed.}
    \label{fig:hybrids}
  \end{figure}

\begin{theorem}
  \label{thm:sec}
 Let $\chan$ be a memoryless channel, $\algp$ be a \ac{PRP} 
relative to $\chan$, the algorithm $\algh$ be a \ac{CRHF}
relative to $\chan$, the cryptosystem $\PKES^{\WOR}$ be the  
cryptosystem designed in Section~\ref{sec:bias} with very sparse support 
relative to $\chan$, and $\mathcal{G}$ be a strongly $2$-universal hash family. 
  The stegosystem $\PKSteg^{*}$ is
  \ac{SS-CCA}-secure against \emph{every} memoryless channel.
\end{theorem}

\begin{proof}[Proof sketch]
  We prove that the above construction is secure via a \emph{hybrid
    argument}. We thus define  six distributions
  $H_{1},\ldots,H_{6}$ shown in Figure~\ref{fig:hybrids}.

  We now proceed by proving that $H_{i}$ and $H_{i+1}$ are
  \ac{SS-CCA}-indistinguishable (denoted by $H_{i}\sim H_{i+1}$). Informally,
  this means that we replace in $\SSCCA-Dist$ the call to the stegosystem (if
  $b=0$) by $H_i$ and the call to the channel (if $b=1$) by $H_{i+1}$. 
 Denote by  $\Adv^{(i)}_{\Ward}(\kappa)$ the advantage of a warden $\Ward$ 
in this situation. 
Clearly, 
  the \ac{SS-CCA}-advantage of $W$ is bounded as $\Adv^{\sscca}_{\Ward,\PKSteg^{*},\chan}(\kappa)\leq
  \Adv^{(1)}_{\Ward}(\kappa)+\Adv^{(2)}_{\Ward}(\kappa)+\Adv^{(3)}_{\Ward}(\kappa)+\Adv^{(4)}_{\Ward}(\kappa)+\Adv^{(5)}_{\Ward}(\kappa)$. This implies the  theorem, as $H_{1}$ simply describes the channel and $H_{6}$ describes  the stegosystem. 
  Infor\-mal\-ly, we argue that:
  \begin{enumerate}
  \item  $H_{1}\sim H_{2}$ because a uniform random permutation on
    a memoryless channel does not change any probabilities;
  \item $H_{2}\sim H_{3}$  because our choice of $b_{1},\ldots,b_{L}$
    and random  permutations  equal the channel by
    Lemma~\ref{lem:infotheoretic};
\item  $H_{3}\sim H_{4}$ because $\algp$ is a \ac{PRP};
\item  $H_{4}\sim H_{5}$ because $\algp$ is a \ac{PRP}; 
\item $H_{5}\sim H_{6}$ because
  $\PKES^{\WOR}$ is secure due to Corollary~\ref{cor:indist:cca:perm} and
  because of Lemma~\ref{lem:produce_negl}. \qed
  \end{enumerate}

\end{proof}

\section{An Impossibility Result}
\label{sec:impossibility}
We first describe an argument for truly random channels using an
infeasible assumption and then proceed to modify those channels to
get rid of this. All channels will be
$0$-memoryless and we thus write $\chan_{\eta,\dl}$ 
instead of $\chan_{\hist,\dl}$, if $\hist$ contains $\eta$ document.

The main idea of our construction lies on the \emph{unpredictability} of
random channels. If $\chan_{\eta}$ and $\chan_{\eta+1}$ are independent and
sufficiently random, we can not deduce anything about $\chan_{\eta+1}$
before we have sampling access to it, which we only have \emph{after} we sent
the document of $\chan_{\eta}$ in the standard non-look-ahead model. To
be reliable, there must be enough
documents in $\chan_{\eta+1}$ continuing the already
sent documents (call those documents \emph{suitable}). 
To be \ac{SS-CCA}-secure, the number of suitable documents
in $\chan_{\eta+1}$ must be very small to prevent replay attacks like
those in Section \ref{sec:detect}. By replacing the random channels with
pseudorandom ones, we
can thus prove that \emph{every} stegosystem is either unreliable or
\ac{SS-CCA}-insecure on one of those channels. To improve the
readability, fix some stegosystem $\Steg$ and let 
$n=\combine{\Steg}{\dl}(\kappa)$ and
$\ell=\combine{\Steg}{\outl}(\kappa)$.

\vspace*{-2mm}
\subsubsection*{Lower Bound on Truly Random Channels.}
\label{sec:lower-bound-totally}
For $n\in \mathbb{N}$, we denote by $\mathcal{R}_{n}$ all subsets
$R$ of
$\{0,1\}^{n}$ such that there is a negligible function $\negl$ with
\begin{itemize}
  \item $|R|\geq \negl(n)^{-1}$ and 
  \item $|R|\leq 2^{n/2}$.
  \end{itemize}
This means each subset $R$ has super-polynomial cardinality in $n$ without
being too large.
For an infinite sequence
$\vec{R}=R_{0},R_{1},\ldots$ with $R_{i}\in \mathcal{R}_{n}$, we
construct a channel $\chan(\vec{R})$ where the distribution
$\chan(\vec{R})_{i,n}$ is the uniform distribution on
$R_{i}$. The family of all such channels is denoted by
$\cF(\mathcal{R}_{n})$. We assume that a warden can test whether
a document $d$ belongs to the support of
$\chan(\vec{R})_{i,n}$ and denote this warden by $\Ward_{\vec{R}}$. In the next section, we replace the
totally random channels by pseudorandom ones and will get rid 
of this infeasible assumption.
For a stegosystem $\Steg$~--~like the system $\PKSteg^{*}$ from
the last section~--~ we are now interested in
two possible events that may occur during the run of $\combine{\Steg}{\Enc}$ on a channel
$\chan(\vec{R})$. The first  event, denoted by $\event_{\Nq}$ (for
\emph{\underline{N}on\underline{q}ueried}), happens if $\combine{\Steg}{\Enc}$ outputs a
document it has not seen due to sampling. We are also interested in the
case that $\combine{\Steg}{\Enc}$ outputs something  
in the support of the
channel, denoted by $\event_{\Ins}$ for \emph{\underline{In}
  \underline{S}upport}.  Clearly, upon random choice of
$\vec{R}$, $\eta$ (the length of the history), $m$ and $\pk$
we have
\begin{align*}
  &\Pr[\event_{\Ins} \mid \event_{\Nq}]  \  \leq\ell\cdot
\frac{2^{n/2}-\combine{\Steg}{\query}(\kappa)}{2^{n}-
  \combine{\Steg}{\query}(\kappa)} \  \leq 
    \ell\cdot 2^{-n/2},
\end{align*}
where   $\combine{\Steg}{\query}(\kappa)$ denotes the number of
queries performed by $\Steg$. 
This is negligible in $\kappa$ as $n$, $\query$ and $\ell$ are polynomials in
$\kappa$. As warden $\Ward_{\vec{R}}$ can 
test
whether a document belongs to the random sets, we have
$\Adv^{\sscca}_{\Ward_{\vec{R}},\Steg,\chan(\vec{R})}(\kappa)\geq
\Pr[\overline{\event_{\Ins}}]$.  Clearly, since we can assume $\overline{\event_{\Ins}}\subseteq \event_{\Nq}$ 
we thus obtain 
$$  \Pr[\event_{\Nq}] \  = \  \frac{\Pr[\overline{\event_{\Ins}} \wedge \event_{\Nq}]}{\Pr[\overline{\event_{\Ins}} \mid \event_{\Nq}]} \  \leq \
  \frac{\Adv^{\sscca}_{\Ward_{\vec{R}},\Steg,\chan(\vec{R})}(\kappa)}{1-\ell\cdot
  2^{-n/2}}.$$
Hence, if $\Steg$ is \ac{SS-CCA}-secure, the term $\Pr[\event_{\Nq}]$ must be negligible.

If $\Steg$ is given a history of length $\eta$ and it outputs 
documents
$d_{1},\ldots,d_{\ell}$, we note that $\combine{\Steg}{\Enc}$ only gets sampling
access to $\chan(\vec{R})_{\eta+\ell-1,n}$ after
it sent $d_{1},\ldots,d_{\ell-1}$ in the standard non-look-ahead model. Clearly, due to the random choice of
$\vec{R}$, the set $R_{\eta+\ell}$ is independent of
$R_{\eta},R_{\eta+1},\ldots,R_{\eta+\ell-1}$. The encoder
$\combine{\Steg}{\Enc}$ thus needs to decide on the documents $d_{1},\ldots,d_{\ell-1}$
without any knowledge of $R_{\eta+\ell}$. As $\combine{\Steg}{\Enc}$ draws a sample
set $D$ from $\chan(\vec{R})_{\eta+\ell-1,n}$
with at most $q=\combine{\Steg}{\query}(\kappa)$ documents, we now look at the event $\event_{\Nsui}$
(for \emph{\underline{N}ot \underline{sui}table}) that none of the
documents in $D$ are suitable for the encoding,
\ie if the sequence $d_{1},d_{2},\ldots,d_{\ell-1},d$ is not a suitable
encoding of the message $m$ for all $d\in D$. Denote the unreliabiliy of the stegosystem by
$\rho$. Clearly, if $\event_{\Nsui}$ occurs, there are two possibilities for the
stegosystem: It  either outputs something from $D$ and thus increases the
unreliability or it  outputs something it has not queried. We thus have
  $\Pr[\event_{\Nsui}] \  \leq \  \max\{\rho, (1-\rho)\cdot \Pr[\event_{\Nq}]\}$.
Note that  $\rho$ must be negligible if $\combine{\Steg}{\Enc}$ is reliable and,
as discussed above, the term $\Pr[\event_{\Nq}]$ (and thus the term $(1-\rho)\cdot
\Pr[\event_{\Nq}]$) must be negligible if $\combine{\Steg}{\Enc}$ is \ac{SS-CCA}-secure. Hence, if $\combine{\Steg}{\Enc}$ is
\ac{SS-CCA}-secure and reliable, the probability $\Pr[\event_{\Nsui}]$ must be
negligible. 
The insight, that $\Pr[\event_{\Nsui}]$ must be negligible directly leads us to
the construction of a warden $\Ward_{\vec{R}}$ on the channel
$\chan(\vec{R})$. The warden chooses a random history
of length $\eta$ and a random message $m$ and sends those to its challenging
oracle. It then receives the document sequence
$d_{1},\ldots,d_{\ell}$. If $d_{i}\not\in R_{\eta+i}$, the warden
returns »Stego«. Else, it samples $q$ documents $D$ from
$\chan(\vec{R})_{\eta+\ell,n}$ and tests for all $d\in D$ via
the decoding oracle $\combine{\Steg}{\Dec}_{\sk}$ if the sequence $d_{1},d_{2},\ldots,d_{\ell-1},d$ decodes
to $m$. If we find such a $d$,  return »Stego« and else
return »Not Stego«. If the documents are randomly chosen from
the channel, the probability to return »Stego« is at most
$q/|2^{\combine{\Steg}{\ml}(\kappa)}|$, \ie negligible. If the documents
are chosen by the stegosystem, the probability of
»Not Stego« is exactly $\Pr[\event_{\Nsui}]$. Hence, $\Steg$ must be
either unreliable or \ac{SS-CCA}-insecure on some channel in $\cF(\mathcal{R}_{n})$. 

\vspace*{-1mm}
\subsubsection*{Lower Bound on Pseudorandom Channels.}
To give a 
proof, we will replace the random channels
$\chan(\vec{R})$ by pseudorandom ones.
The construction assumes existence of a
\ac{CCA+}-secure cryptosystem $\PKES$ with
$\combine{\PKES}{\cl}(\kappa)\geq 2\combine{\PKES}{\ml}(\kappa)$.

For $\omega=(\pk,\sk)\in
\supp(\combine{\PKES}{\Gen}(1^{\kappa}))$, let $\chan(\omega)_{i,\dl(\kappa)}$ be the
distribution $\combine{\PKES}{\Enc}(\pk,\binary(i)_{\dl(\kappa)})$,
where $\binary(i)_{\dl(\kappa)}$ is the
binary representation of the number $i$ of length exactly $\dl(\kappa)$ modulo $2^{\dl(\kappa)}$. The
family of channels $\bm{\mathcal{C}}_{\PKES}=\{\chan(\omega)\}_{\omega}$ thus has the following
properties:

\begin{enumerate}
\item There is a negligible function $\negl$ such that for each $\omega$
  and each $i$, we have
  $2^{\combine{\PKES}{\ml}(\kappa)/2}\geq
  |\chan(\omega)_{i,\dl(\kappa)}|\geq \negl(\kappa)^{-1}$ if $\PKES$ is
  \ac{CCA+}-secure. This follows easily from the \ac{CCA+}-security of
  $\PKES$: If $|\chan(\omega)_{i,\dl(\kappa)}|$ would be polynomial, an
  attacker could simply store all corresponding ciphertexts. 
\item An algorithm with the knowledge of $\omega$ can test in polynomial
  time, whether $d\in \supp(\chan(\omega)_{i,\dl(\kappa)})$, \ie whether $d$
  belongs to the support by simply testing whether
  $\combine{\PKES}{\Dec}(\sk,d)$ equals $\binary(i)_{\dl(\kappa)}$. 
\item Every algorithm $\algo{Q}$ that tries to distinguish $\chan(\omega)$ from a
  random channel
  $\chan(\vec{R})$ fails: For every polynomial algorithm $\algo{Q}$, we
  have that the term
  \begin{align*}
    \bigl|
    &\Pr_{\vec{R}\sgets \bm{\mathcal{R}}_{\dl(\kappa)}^{*}}[\algo{Q}^{\chan(\vec{R})}(1^{\kappa})=1]-
    \Pr_{\omega\gets \combine{\PKES}{\Gen}(1^{\kappa})}[\algo{Q}^{\chan(\omega)}(1^{\kappa})=1]\bigl|
  \end{align*}
  is negligible in $\kappa$ if $\PKES$ is \ac{CCA+}-secure. This follows
  from the fact that no polynomial algorithm can distinguish
  $\chan(\vec{R})$ upon random choice of $\vec{R}$ from the uniform
  distribution on $\{0,1\}^{n}$, as $|\chan(\vec{R})_{i,n}|$ is
  super-polynomial. Furthermore,  an attacker
  $\Att$ on $\PKES$ can simulate $\algo{Q}$ for a
  successful attack. 
\end{enumerate}

Note that the third property directly implies that no polynomial
algorithm can conclude anything about $\chan(\omega)_{i,\dl(\kappa)}$
from samples of previous distributions
$\chan(\omega)_{0,\dl(\kappa)}, \ldots
,\chan(\omega)_{i-1,\dl(\kappa)}$, except for a negligible term. The
second property directly implies that we can get rid of the infeasible
assumption of the previous section concerning the ability of the warden
to test whether a document belongs to the support of $\chan(\omega)$: We
simply equip the warden with the seed $\omega$. Call the resulting
warden $\Ward_{\omega}$. Note that this results in a non-uniform warden.
As above, we are interested in the events that a stegosystem outputs a
document that it has not seen
($\event_{\widehat{\Nq}}$), 
that a document is outputted which 
does not belong to the support ($\event_{\widehat{\Ins}}$) and
the event that a random set of $q$ documents is not suitable to
complete a given document prefix
$d_{1},d_{2},\ldots,d_{\ell-1}$
($\event_{\widehat{\Nsui}}$).

As $\event_{\widehat{\Ins}}$ is a polynomially testable property (due to the
second property of our construction), we can conclude a similar bound as above:
\begin{lemma}
\label{lem:dedic}
  Let $\Steg$ be an \ac{SS-CCA}-secure universal stegosystem. For every
  warden $\Ward$ and every \acs{CCA+}-attacker $\Att$,\  
$\Pr[\event_{\widehat{\Nq}}] \  \leq \   \frac{\Adv^{\sscca}_{\Ward,\Steg,\chan(\omega)}(\kappa)}{1-\ell\cdot
  2^{-n/2}}+\Adv^{\pkess}_{\Att,\PKES}(\kappa)$.
\end{lemma}
Hence, if  the stegosystem $\Steg$ is \ac{SS-CCA}-secure and $\PKES$
is \acs{CCA+}-secure, the term $\Pr[\event_{\widehat{\Nq}}]$ must be negligible. 
As above,  we can conclude that
  $\Pr[\event_{\widehat{\Nsui}}] \ \leq \ \max\{\rho, (1-\rho)\cdot
  \Pr[\event_{\widehat{\Nq}}]\}$ for unreliabiliy $\rho$.
The warden $\Ward_{\omega}$ similar to
$\Ward_{\vec{R}}$ from the preceding section thus suceeds with
very high probability. Hence, no \ac{SS-CCA}-secure and reliable stegosystem can
exist for the family  $\bm{\mathcal{C}}_{\PKES}$:

\begin{theorem}\label{thm:impossibility}
  If doubly-enhanced trapdoor permutations exist, for every stegosystem $\Steg$ in the
  non-look-ahead model there is a $0$-memoryless channel $\chan$ such
  that $\Steg$ is either unreliable or it is not \ac{SS-CCA}-secure
  on~$\chan$ against non-uniform wardens.
\end{theorem}

\section{Discusssion} 
The work of Dedi{\'c} et al.~\cite{dedic2009upper} shows that provable secure
universal steganography needs a huge number of sample documents to embed long
secret messages as high bandwidth stegosystems are
needed for such messages. This limitation also transfers to the public-key
scenario. However, such a limitation does not necessarily restrict applicability
of steganography, especially in case of specific communication channels or
if the length of secret messages is sufficiently short.

A prominent recent example for such applications is the use of steganography for channels
determined by cryptographic primitives, like \acp{SES} or digital signature
schemes. Bellare, Paterson, and Rogaway introduced in \cite{bellare2014asa} so
called \emph{algorithm substitution attacks} against \acp{SES}, where an
attacker replaces an honest implementation of the encryption algorithm by a
modified version which allows to extract the secret key from the ciphertexts
produced by the corrupted implementation. Several follow-up works have been done
based on this paper, such as those by Bellare, Jaeger, and Kane
\cite{bellare2015asa}, Ateniese, Magri, and Ventur \cite{ateniese2015sig}, or
Degabriele, Farshim, and Poettering \cite{degabriele2015asa}. These works
strengthened the model proposed in \cite{bellare2014asa} and presented new
attacks against \acp{SES} or against other cryptographic primitives, \eg against signature schemes. 
Surprisingly, as we show in 
\cite{berndt2017asa}, all such algorithm
substitution attacks can be analyzed in the framework of computational
secret-key steganography and in consequence, the attackers can be identified as stegosystems 
on certain channels determined by the primitives.
In such scenarios, the secret message embedded by the stegosystem corresponds
to  a secret key of the cryptographic algorithm.

A similar approach was used by Pasquini, Schöttle, and Böhme
\cite{pasquini2017decoy} to show that so called \emph{password decoy vaults}
used for example by Chatterjee, Bonneau, Juels, and Ristenpart
\cite{chatterjee2015decoy} and Golla, Beuscher, and Dürmuth
\cite{golla2016decoy} can also be interpreted as steganographic protocols.

\bibliography{lib}

\appendix

\section{Remaining Proofs}
\label{app:proofs:att:generate}

To improve the readability, we will abbreviate some terms and define
$n=\combine{\PKSteg^{*}}{\dl}(\kappa)$,
$\ell=\combine{\PKSteg^{*}}{\outl}(\kappa)$ and
$L=\combine{\PKES^{\WOR}\!}{\cl}(\kappa)$, where $\PKSteg^{*}$ is our
stegosystem constructed in Section~\ref{sec:stegosystem} and
$\PKES^{\WOR}$ is the public-key cryptosystem constructed in
Section~\ref{sec:bias}. We also define $N=8L$.

\subsection{Formal Statement of  Lemma~\ref{lem:produce_negl} and its Proof}
We start with a formal definition for ``$\algo{A}$ is successful on
$D,f,b_{1},\ldots,b_{L},k_{\algh}$''.
\begin{definition}
  \label{def:attack}
  An attacker $\algo{A}$ on $\algo{generate}$ is a \ac{PPTM}, that receives the following input: 
  \begin{compactitem}
  \item a sequence $d_{1},\ldots,d_{N}$ of $N$ pairwise different documents
  \item a hash function $f:\Sigma^n \to \{0,1\}$ from the family
    $\mathcal{G}=\{G_{\kappa}\}_{\kappa\in \mathbb{N}}$, 
  \item a sequence $b_{1},\ldots,b_{L}$ of $L$ bits, and 
  \item a hash-key $k_{\algh}$ for $\algh$.
  \end{compactitem}
  The attacker $\algo{A}$ then outputs a sequence  $d'_{1},\ldots,d'_{N}$ of documents.  
  Note that the attacker  knows the mapping function $f$ and even the 
  hash-key $k_{H}$ for $\algh$.

  We say that $\algo{A}$ is \emph{successful} if the experiment
  $\Sgen(\algo{A},D,f,b_{1},\ldots,b_{L})$ returns value $1$:

  \myAlgorithm{$\Sgen(\algo{A},D,f,b_{1},\ldots,b_{L})$}{
    Attacker $\algo{A}$, set $D$, function $f$, bits $b_{1},\ldots,b_{L}$}{
    \State $k_{\algp}, k'_{\algp} \gets \combine{\algp}{\Gen}(1^{\kappa})$
    \State $k_{\algh} \gets \combine{\algh}{\Gen}(1^{\kappa})$
    \State $d_{1},\ldots,d_{N} :=
    \algo{generate}(D,f,b_{1},\ldots,b_{L},k_{\algp},k'_{\algp})$
\State $d'_{1},\ldots,d'_{N} \gets \algo{A}(d_{1},\ldots,d_{N},f,
b_{1},\ldots,b_{L},k_{\algh})$
\If{$f(d'_{i})=b_{i}$ for every $i=1,\ldots L$}
\State $D'_{0}=\{d'_{j}\mid f(d'_{j})=0\}$; $D'_{1}=\{d'_{j}\mid
f(d'_{j})=1\}$
\If{
    $d'_{1},\ldots,d'_{N}=\algo{generate}(D'_{0}\cup D'_{1},f,b_{1},\ldots,b_{L},k_{\algp},k'_{\algp})$}
\If{
  $\combine{\algh}{\Eval}_{k_{\algh}}(\lex(D'_{0}\cup D'_{1}))=\combine{\algh}{\Eval}_{k_{\algh}}(\lex(D_{0}\cup
  D_{1}))$}
\If{$d'_{1},\ldots,d'_{N}\neq d_{1},\ldots,d_{N}$}
\State \AlgReturn{$1$} and {\bf halt}
\EndIf
\EndIf
\EndIf
\EndIf
\State\AlgReturn{$0$}
} {Security of $\algo{generate}$}
\end{definition}

We are now ready to give the formal version of Lemma~\ref{lem:produce_negl}:
\begin{lemma*}[formal version of Lemma~\ref{lem:produce_negl}]
  Let $D\subseteq \Sigma^n$ be a set of documents, with 
  $|D|=N$, let  $b_{1},\ldots,b_{L}$ be a bitstring, and $f\in G_{\kappa}$.  
  For every attacker $\algo{A}$ on $\algo{generate}$, there is a
  collision finder $\Fi$ for the \ac{CRHF} $\algh$ 
  such that
  \begin{align*}
    \Pr[\Sgen(\algo{A},D,f,b_{1},\ldots,b_{L})=1]\leq  \Adv_{\Fi,\algh,\chan}^{\hash}(\kappa),
  \end{align*}
  where the probability is taken over the random choices made in
  experiment $\Sgen$.
\end{lemma*}
\begin{proof}
  Let $\algo{A}$ be an attacker on $\algo{generate}$ with maximal success
  probability. Let $D=D_{0}\dot{\cup} D_{1}$ be the input to
  $\algo{generate}$, the sequence $d_{1},\ldots,d_{N}$ its output and $d'_{1},\ldots,d'_{N}$ be the output of
  $\algo{A}$. Furthermore, let $D'_{b}=\{d'_{j}\mid f(d'_{j})=b\}$ and
  $D'=D'_{0}\cup D'_{1}$. We now distinguish three cases of the relation between $D$ and
  $D'$. If $D'\subsetneq D$, the sequence $d'_{1},\ldots,d'_{N}$ must
  contain the same element on at least two positions, but
  $\algo{generate}$ does only accept sets of size exactly $N$. Hence,
  $\algo{A}$ was not successful in this case. If $D'=D$ and $\algo{A}$ was successful,
  it holds that $d'_{1},\ldots,d'_{N}\neq d_{1},\ldots,d_{N}$. Hence,
  there must be positions $i < j$ and $j' < i'$ such that
  $d_{i}=d_{i'}$ and $d_{j}=d_{j'}$. As $k_{\algp}$ and $k'_{\algp}$ define
  a total order, the sequence $d'_{1},\ldots,d'_{N}$
  could not be produced by $\algo{generate}$. Thus, $\algo{A}$ can not be
  successful in this case. 

The only remaining case is $D'\setminus D\neq \emptyset$.  If $\algo{A}$ was
  successful, it holds that
  $\algh_{k_{\algh}}(\lex(D'))=\algh_{k_{\algh}}(\lex(D))$, \ie this is a
  collision with regard to $\algh$. We will now construct a finder
  $\Fi$ for $\algh$, such that
  $\Adv^{\hash}_{\Fi,\algh,\chan}(\kappa)\geq \Pr[\text{$\algo{A}$ succeeds}]$.  The
  finder $\Fi$ receives a hash key $k_{\algh}$. It then chooses
  $f\sgets G_{\kappa}$, samples $D$ documents of cardinality $|D|=N$ via rejection sampling
  and  \ac{PRP}-keys $k_{\algp},k'_{\algp}$. The finder simulates 
  $\algo{A}$ and receives
  \begin{align*}
    d'_{1},\ldots,d'_{N}\gets
  \algo{A}(\algo{generate}(D,f,b_{1},\ldots,b_{L},k_{\algp},k'_{\algp}),f,b_{1},\ldots,b_{L},k_{\algh}).
  \end{align*}

  Then, it returns $D$ and $D'=\{d'_{1},\ldots,d'_{N}\}$. Whenever $\algo{A}$
  succeeds, we have $D\neq D'$ by the discussion above and thus also 
  $\algh_{k_{\algh}}(\lex(D))=\algh_{k_{\algh}}(\lex(D'))$. Hence, $\Fi$ has
  successfully found a collision. This implies that
  $\Adv^{\hash}_{\Fi,\algh,\chan}(\kappa)\geq \Pr[\text{$\algo{A}$ succeeds}]$. 
\qed\end{proof}
\subsection{Proof of Theorem~\ref{thm:rel}}
Recall the statement of the theorem:
\begin{theorem*}[Theorem~\ref{thm:rel}]
  The probability that a message is not correctly embedded by $\combine{\PKSteg^{*}}{\Enc}$ 
  is at most
  $3N^{2}\cdot 2^{-\minent(\chan,\kappa)}+2\exp(-N/54)$.
\end{theorem*}
\begin{proof}
  Note that $\combine{\PKSteg^{*}}{\Enc}$ may not correctly embed a message $m$ if 
  (a) $|D_{0}\cup D_{1}|< N$ \ie a document sampled in line~3 was drawn twice, or 
  (b) $N_{0}/N\not\in [1/3,2/3]$  \ie the bias is too large, or 
  (c) the number of elements of $D_{0}$ or $D_{1}$ is too small to embed
  $b_1,b_2,\ldots,b_{L}$
  by $\algo{generate}$.
  The probability of (a) can be bounded
  similar to the birthday attack. It is roughly bounded by
  $3N^{2}\cdot 2^{-\minent(\chan,\kappa)}$ as the probability of every
  document is bounded by $2^{-\minent(\chan,\kappa)}$.

  A simple calculation shows that the probability of (b) and (c) is
  negligible. Note that the algorithm always correctly embeds a message,
  if $|D_{0}|\geq L$ and $|D_{1}|\geq L$. As
  $N_{0}/N=|D_{0}|/N$, this implies that
  $N_{0}/N\in [1/3,2/3]$. 
  We will thus estimate the
  probability for this. As $f$ is drawn from a strongly $2$-universal
  hash family, we note that the probability that a random document $d$ is
  mapped to $1$ is equal to $1/2$. For $i=1,\ldots,N$, let $X_{i}$ be
  the indicator variable such that $X_{i}$ equals $1$ if the $i$-th element
  drawn from the channel maps to $1$. The random variable
  $X=\sum_{i=1}^{N}X_{i}$ thus has the size of
  $D_{1}$. Clearly, its expected value is $N/2$. The probability that
  $|X-N/2| > L$  (and thus $|D_{1}| < L$ or $|D_{0}| < L$) is hence
  bounded by
  \begin{align*}
    \Pr[|X-N/2| > L]\leq 2\exp(-\frac{L \cdot (1/3)^{2}}{3})=2\exp(-N/54)
  \end{align*}
  using  a Chernoff-like bound.
  The probability that the message $m$ is incorrectly
  embedded is thus bounded by
  $2^{-\minent(\chan,\kappa)}+2\exp(-N/54)$.
\qed\end{proof}

\subsection{Proof of Theorem~\ref{thm:sec}}
We recall:
\begin{theorem*}[Theorem~\ref{thm:sec}]
 Let $\chan$ be a memoryless channel, $\algp$ be a \ac{PRP} 
relative to $\chan$, the algorithm $\algh$ be a \ac{CRHF}
relative to $\chan$, the cryptosystem $\PKES^{\WOR}$ be the  
cryptosystem designed in Section \ref{sec:bias} with very sparse support 
relative to $\chan$, and $\mathcal{G}$ be a strongly $2$-universal hash family. 
  The stegosystem $\PKSteg^{*}$ is
  \ac{SS-CCA}-secure against \emph{every} memoryless channel.
\end{theorem*}

\begin{proof}
    We prove that the above construction is secure via a \emph{hybrid
    argument}. We thus define  six distributions
  $H_{1},\ldots,H_{6}$ shown in Figure~\ref{fig:hybrids}.

  If $P$ and $Q$ are two probability distributions, denote by $\SSCCA-Dist_{P,Q}$
  the modification of the game  $\SSCCA-Dist$, where the call to the stegosystem
  (if $b=0$) is replaced by a  call to $P$ and the call to the channel (if
  $b=1$) is replaced by a call to  $Q$. If $\Ward$ is some warden,
  denote by  $\Adv^{\sscca}_{\Ward,P,Q}(\kappa)$ the winning probability of
  $\Ward$ in $\SSCCA-Dist_{P,Q}$. If $\Adv^{\sscca}_{\Ward,P,Q}(\kappa)\leq
  \negl(\kappa)$ for a negligible function $\negl$, we denote this situation as
  $P\sim Q$ and say that $P$ and $Q$ are \emph{indistinguishable} with respect
  to $\SSCCA-Dist$. Furthermore, we define 
  $\Adv^{(i)}_{\Ward}(\kappa)=\Adv^{\sscca}_{\Ward,H_i,H_{i+1}}(\kappa)$. As
  the term $\Adv^{(i)}_{\Ward}(\kappa)$ can also be written as
  \begin{align*}
\bigl|\Pr[\text{$\combine{\Ward}{\Guess}$ outputs $b'=0$} \mid  b=0]-\Pr[\text{$\combine{\Ward}{\Guess}$ outputs $b'=0$} \mid b=1]\bigr|,    
  \end{align*}
 the triangle inequality implies that 
  $\Adv^{\sscca}_{\Ward,\PKSteg^{*},\chan}(\kappa)\leq  \Adv^{(1)}_{\Ward}(\kappa)+\Adv^{(2)}_{\Ward}(\kappa)+\Adv^{(3)}_{\Ward}(\kappa)+\Adv^{(4)}_{\Ward}(\kappa)+\Adv^{(5)}_{\Ward}(\kappa)$.    
  
  Infor\-mal\-ly, we argue that:
  \begin{enumerate}
  \item  $H_1=H_2\implies H_{1}\sim H_{2}$ because a uniform random permutation on
    a memoryless channel does not change any probabilities;
  \item $H_2=H_3\implies H_{2}\sim H_{3}$  because our choice of $b_{1},\ldots,b_{L}$
    and random  permutations  equal the channel by
    Lemma~\ref{lem:infotheoretic};
\item  $H_{3}\sim H_{4}$ because $\algp$ is a \ac{PRP};
\item  $H_{4}\sim H_{5}$ because $\algp$ is a \ac{PRP}; 
\item $H_{5}\sim H_{6}$
  $\PKES^{\WOR}$ is secure due to Corollary~\ref{cor:indist:cca:perm} and
    because of Lemma~\ref{lem:produce_negl}. 
  \end{enumerate}
Distribution $H_1$ can be specified as follows:
\begin{center}
      \begin{minipage}{.48\linewidth}
    \begin{gameproof}[name=$H$,arg={=\chan^{N}_n}]
      \gameprocedure[linenumbering]{
        \pk^{*}=(\pk,f) \gets \combine{\PKSteg^{*}}{\Gen}(1^{\kappa})\\
      \pcfor j := 1,2,\ldots,N:\\
      \t d_{j} \gets \chan_{\dl(\kappa)}\\
      \pcreturn ((d_{1},\ldots,d_{N}),\pk^{*})
    }
\end{gameproof}
      \end{minipage}
\end{center}

    \begin{description}
  \item[Indistinguishability of $H_1$ and]{$ $}
\begin{center}
      \begin{minipage}{.48\linewidth}
\begin{gameproof}[name=$H$, arg={}, nr=1] 
    \gameprocedure{
      \pk^{*}=(\pk,f) \gets \combine{\PKSteg^{*}}{\Gen}(1^{\kappa})\\
      \text{Lines 1 to 4 in
        $\combine{\PKSteg^{*}}{\Enc}$}\\
      \setcounter{pclinenumber}{4}
      \pcln \gamechange{$P\sgets \Perm$}\\
      \pcln  \pcreturn ((\gamechange{$d_{P(1)},\ldots,d_{P(N)}$}),\pk^{*})
    }
  \end{gameproof}
      \end{minipage}\\
\end{center}
     If $|D_{0}\cup D_{1}|< N$, \ie a
    document was sampled twice or $|D_{0}|/|D| \not\in [1/3,2/3]$, the
    system only outputs the sampled documents. Hence $H_{1}$ equals
    $H_{2}$ in this case. In the other case, we first permute the items
    before we output them. But, as $P$ is a uniform random
    permutation and the documents are drawn  independently from a
    memoryless channel, we have 
    \begin{align*}
    \Pr_{H_{1}}[\text{$d_{1},\ldots,d_{N}$
      are drawn}]=\Pr_{H_{1}}[\text{$d_{P(1)},\ldots,d_{P(N)}$
      are drawn}].      
    \end{align*}
    As $\pk$ is not used in these hybrids, $H_{1}=H_{2}$
    follows. \\
  \item[Indistinguishability of $H_2$ and]{$ $} 
\begin{center}
  \begin{gameproof}[name=$H$, arg={}, nr=2]
        \gameprocedure{
      \pk^{*}=(\pk,f) \gets \combine{\PKSteg^{*}}{\Gen}(1^{\kappa})\\
      \text{Lines 1 to 4 in
        $\combine{\PKSteg^{*}}{\Enc}$}\\
      \setcounter{pclinenumber}{4}
      \pcln P\sgets \Perm; \gamechange{$P'\sgets \Perm$}; \gamechange{$k_{\algh}\gets \combine{\algh}{\Gen}(1^{\kappa})$}\\
      \pcln \gamechange{$b_{1},b_{2},\ldots, b_{L} \gets
        D^{\WOR}_{(N,N_0,L)}$}\\
      \pcln \pcreturn (\gamechange{$
      \algo{generate}(D_{0}\cup D_{1},f,b_{1},\ldots,b_{L},P,P')$},\pk^{*})\\
      \pccomment{$\algo{generate}(\ldots,P,P')$ uses the permutations $P,P'$}
    }
  \end{gameproof}
\end{center}
   If
    $|D_{0}\cup D_{1}|< N$, \ie a document was sampled twice or
    $|D_{0}|/|D| \not\in [1/3,2/3]$, the system only outputs the sampled
    documents. Hence $H_{2}$ equals $H_{3}$ in this case.  If
    $|D_{0}\cup D_{1}|=N$, Lemma~\ref{lem:infotheoretic} shows that
    $H_{2}$ equals $H_{3}$.\\
  \item[Indistinguishability of $H_3$ and] {$ $}
\begin{center}
  \begin{gameproof}[name=$H$, arg={},nr=3]
    \gameprocedure{
      \pk^{*}=(\pk,f) \gets \combine{\PKSteg^{*}}{\Gen}(1^{\kappa})\\
      \text{Lines 1 to 4 in
        $\combine{\PKSteg^{*}}{\Enc}$}\\
      \setcounter{pclinenumber}{4}
      \pcln \gamechange{$k_{\algp} \gets
        \combine{\algp}{\Gen}(1^{\kappa})$}; P'\sgets \Perm; k_{\algh}\gets \combine{\algh}{\Gen}(1^{\kappa})\\
      \pcln b_{1},b_{2},\ldots, b_{L} \sgets D^{\WOR}_{(N,N_0,L)}\\
      \pcln \pcreturn
      (\algo{generate}(D_{0}\cup
      D_{1},f,b_{1},\ldots,b_{L},\gamechange{$k_{\algp}$},P'),\pk^{*})\\
      \pccomment{$\algo{generate}(\ldots,P')$ uses the permutation $P'$}
      }
    \end{gameproof}
\end{center}
    We will construct a distinguisher $\Dist$ on the \ac{PRP} $\algp$ with
    $\Adv^{\prp}_{\Dist,\algp,\chan}(\kappa)=\Adv^{(3)}_{\Ward}(\kappa)$. Note that
    such a distinguisher has access to an oracle  that either
    corresponds to a truly random permutation or to
    $\combine{\algp}{\Eval}_{k}$ for a    
    key~$k\gets \combine{\algp}{\Gen}(1^{\kappa})$. 

    The \ac{PRP}-distinguisher $\Dist$  simulates the run of $\Ward$. It
    first chooses a key-pair $(\pk,\sk)\gets
    \combine{\PKSteg^{*}}{\Gen}(1^{\kappa})$. It
    then simulates $\Ward$. Whenever the warden $\Ward$ makes a
    call to its decoding-oracle $\combine{\PKSteg^{*}}{\Dec}$,
    it computes $\combine{\PKSteg^{*}}{\Dec}(\sk,\cdot)$ (or $\bot$ if
    necessary). In order to generate the challenge sequence $\hat{d}$ upon
    the message $m$, it simulates the run of $\combine{\PKSteg^{*}}{\Enc}$ and replaces every
    call to $P$ or $\combine{\algp}{\Eval}_{k_{\algp}}$ by a call to its oracle.
    Similarly, the bits output by $\combine{\PKES^{\WOR}\!}{\Enc}(\pk,m)$ are ignored and
    replaced by truly random bits distributed according to
    $D^{\WOR}_{(N,|D_{0}|,L)}$. If the oracle
    is a truly random permutation, the
    simulation yields exactly $H_{3}$ and if the oracle equals $\combine{\algp}{\Eval}_{k}$ for a
    certain key $k$, the simulation yields $H_{4}$. The advantage of $\Dist$
    is thus exactly $\Adv^{(3)}_{\Ward}(\kappa)$. As $\algp$ is a secure
    \ac{PRP}, this advantage is negligible and $H_{3}$ and $H_{4}$
    are thus indistinguishable.\\
  \item[Indistinguishability of $H_4$ and] {$ $}
\begin{center}
     \begin{gameproof}[name=$H$, arg={},nr=4]
    \gameprocedure{
      \pk^{*}=(\pk,f) \gets \combine{\PKSteg^{*}}{\Gen}(1^{\kappa})\\
      \text{Lines 1 to 4 in
        $\combine{\PKSteg^{*}}{\Enc}$}\\
      \setcounter{pclinenumber}{4}
      \pcln k_{\algp} \gets
        \combine{\algp}{\Gen}(1^{\kappa}); \gamechange{$k'_{\algp} \gets
        \combine{\algp}{\Gen}(1^{\kappa})$}; k_{\algh}\gets \combine{\algh}{\Gen}(1^{\kappa})\\
      \pcln b_{1},b_{2},\ldots, b_{L} \sgets D^{\WOR}_{(N,N_0,L)}\\
      \pcln \pcreturn (\algo{generate}(D_{0}\cup
      D_{1},f,b_{1},\ldots,b_{L},k_{\algp},\gamechange{$k'_{\algp}$}),\pk^{*})
      }
    \end{gameproof}
\end{center} 
    
    We will construct a distinguisher $\Dist$ on the \ac{PRP} $\algp$ with
    $\Adv^{\prp}_{\Dist,\algp,\chan}(\kappa)=\Adv^{(4)}_{\Ward}(\kappa)$. Note that
    such a distinguisher has access to an oracle  that either
    corresponds to a truly random permutation or to
    $\combine{\algp}{\Eval}_{k}$ for a    
    key~$k\gets \combine{\algp}{\Gen}(1^{\kappa})$. 

    The \ac{PRP}-distinguisher $\Dist$  simulates the run of $\Ward$. It
    first chooses a key-pair $(\pk,\sk)\gets
    \combine{\PKSteg^{*}}{\Gen}(1^{\kappa})$ and a key $k_{\algp}\gets
    \combine{\algp}{\Gen}(1^{\kappa})$ for the \ac{PRP} $\algp$.  It
    then simulates $\Ward$. Whenever the warden $\Ward$ makes a
    call to its decoding-oracle $\combine{\PKSteg^{*}}{\Dec}$,
    it computes $\combine{\PKSteg^{*}}{\Dec}(\sk,\cdot)$ (or $\bot$ if
    necessary). In order to generate the challenge sequence $\hat{d}$ upon
    the message $m$, it simulates the run of $\combine{\PKSteg^{*}}{\Enc}$ and replaces every
    call to $P'$ or $\combine{\algp}{\Eval}_{k_{\algp}}$ by a call to its oracle.
    Similarly, the bits output by $\combine{\PKES^{\WOR}\!}{\Enc}(\pk,m)$ are ignored and
    replaced by truly random bits distributed according to
    $D^{\WOR}_{(N,|D_{0}|,L)}$. If the oracle is a truly random permutation, the
    simulation yields exactly $H_{4}$ and if the oracle equals $\combine{\algp}{\Eval}_{k}$ for a
    certain key $k$, the simulation yields $H_{5}$. The advantage of $\Dist$
    is thus exactly $\Adv^{(4)}_{\Ward}(\kappa)$. As $\algp$ is a secure
    \ac{PRP}, this advantage is negligible and $H_{4}$ and $H_{5}$
    are thus indistinguishable.\\
  \item[Indistinguishability of $H_5$ and] {$ $}
\begin{center}
      \begin{gameproof}[name=$H$, arg={=\combine{\PKSteg^{*}}{\Enc}},nr=5]
        \gameprocedure{
          \pk^{*}=(\pk,f) \gets \combine{\PKSteg^{*}}{\Gen}(1^{\kappa})\\
      \text{Lines 1 to 4 in
        $\combine{\PKSteg^{*}}{\Enc}$}\\
      \setcounter{pclinenumber}{4}
      \pcln k_{\algp} \gets \combine{\algp}{\Gen}(1^{\kappa}); k'_{\algp} \gets
        \combine{\algp}{\Gen}(1^{\kappa}); k_{\algh}\gets
        \combine{\algh}{\Gen}(1^{\kappa})\\
        \pcln \gamechange{$h := \combine{\algh}{\Eval}_{k_{\algh}}(\lex(D_0 \cup
          D_1))$}\\
        \pcln \gamechange{$\combine{\PKES^{\WOR}\!\!}{\setup}(N,N_0)$}\\
      \pcln b_{1},b_{2},\ldots, b_{L} \gets \gamechange{$\combine{\PKES^{\WOR}}{\Enc}(\pk,m\mid\mid k_{\algh}\mid\mid
  k_{\algp}\mid\mid k'_{\algp} \mid\mid h)$}\\
      \pcln \pcreturn (\algo{generate}(D_{0}\cup
      D_{1},f,b_{1},\ldots,b_{L},k_{\algp},k'_{\algp}),\pk^{*})
      }
  \end{gameproof}
\end{center}
    We construct an
    attacker $\Att$ on $\PKES^{\WOR}$ such that there is a negligible function
    $\negl$ with
    $\Adv^{\cca}_{\Att,\PKES^{\WOR},\chan}(\kappa)+\negl(\kappa) \geq
    \Adv^{(5)}_{\Ward}(\kappa)$.
    Note that such an attacker $\Att$ has access to the
    decryption-oracle $\combine{\PKES^{\WOR}\!}{\Dec}_{\sk}(\cdot)$.
    
    The attacker $\Att$ simply simulates $\Ward$. First, it chooses
    $f\sgets G_{\kappa}$. Whenever $\Ward$ uses its decryption-oracle to decrypt
    $d_{1},\ldots,d_{N}$, the attacker $\Att$ simulates
    $\combine{\PKSteg^{*}}{\Dec}(d_{1},\ldots,d_{N})$ and
    uses 
    its own decryption-oracle $\combine{\PKES^{\WOR}\!}{\Dec}_{\sk}(\cdot)$ in this. When $\Ward$ outputs the
    challenge $m$, the attacker $\Att$ chooses all of the parameters 
    $D_{0},D_{1},k_{\algh},k_{\algp},k'_{\algp}$ as in $\combine{\PKSteg^{*}}{\Enc}$ and chooses its own challenge
    $\widetilde{m} := m \mid\mid k_{\algh}\mid\mid k_{\algp}\mid \mid
    k'_{\algp}\mid\mid 
    h$, where $h=\combine{\algh}{\Eval}_{k_{\algh}}(D_{0}\cup D_{1})$.

    The attacker now either receives $\vec{b}\gets \combine{\PKES^{\WOR}\!}{\Enc}(\pk,\widetilde{m})$ or
    $L$ random bits $\vec{b}$  from $D^{\WOR}_{(N,|D_{0}|,L)}$ and computes
    \begin{align*}
d_{1},\ldots,d_{N}=\algo{generate}(D_{0} \cup D_{1},f,b_{1},\ldots,b_{L},k_{\algp},k'_{\algp}).      
    \end{align*}
    If the bits correspond to
    $\combine{\PKES^{\WOR}\!}{\Enc}(\pk,\widetilde{m})$, this simulates the
    stegosystem and thus $H_{6}$ perfectly. If the bits are random, this
    equals~$H_{5}$.

    After the challenge is determined, $\Att$ continues to simulate
    $\Ward$. Whenever $\Ward$ uses its decryption-oracle to decrypt
    $d_{1},\ldots,d_{N}$, it behaves as above. There is now a
    significant difference to the pre-challenge situation: The attacker
    $\Att$ is not allowed to decrypt the bits
    $\vec{b}=b_{1},\ldots,b_{L}$. Hence, when $\Ward$ tries to
    decrypt documents $d_{1},\ldots,d_{N}$ such that
    $f(d_{i})=b_{i}$, it has no way to use its
    decryption-oracle and must simply return $\bot$. Suppose that this
    situation arises. Note that the decryption-oracle of $\Ward$ would only
    return a message not equal to $\bot$ then iff
    $d_{1},\ldots,d_{N}=\algo{generate}(D_{0} \cup D_{1},f,\vec{b},k_{\algp},k'_{\algp})$
    and $\combine{\algh}{\Eval}_{k_{\algh}}(\{d_{1},\ldots,d_{N}\})=h$. 

    If $\vec{b}$ is a truly random string from $D^{\WOR}_{(N,|D_{0}|,L)}$, the
    sparsity of $\PKES^{\WOR}$ implies that
    the probability that $\vec{b}$ is a valid encoding is negligible. Hence
    the probability that the decryption-oracle of $\Ward$ would return a
    message not equal to $ \bot$ is negligible. It only remains to prove
    that the probability that the decryption-oracle of $\Ward$ returns a
    message not equal to $\bot$ is negligible if $\vec{b}$ is a valid
    encryption of a message. But Lemma~\ref{lem:produce_negl} states
    just that. We thus have
    $\Adv^{\cca}_{\Att,\PKES^{\WOR},\chan}(\kappa)+\negl(\kappa) \geq
    \Adv^{(5)}_{\Ward}(\kappa)$.
    As the system $\PKES^{\WOR}$ is \ac{CCA}-secure by Corollary~\ref{cor:indist:cca:perm}, this advantage is negligible.  Hence,
    $H_{5}$ and $H_{6}$ are indistinguishable.
   \end{description}
  Hence, the stegosystem $\PKSteg^{\WOR}$ is \ac{SS-CCA}-secure on $\chan$.
\qed\end{proof}

\end{document}